\documentclass[reqno,a4paper]{article}
\pdfoutput=1
\usepackage[top=1in, bottom=1.25in, left=1in, right=1in]{geometry}
\usepackage[british]{babel}
\usepackage[utf8]{inputenc}

\usepackage[T1]{fontenc}
\usepackage{lmodern}
\usepackage{amssymb}

\usepackage{pifont}
\newcommand{\cmark}{\ding{51}}%
\newcommand{\xmark}{\ding{55}}%
\usepackage{amsmath,amsthm,amsfonts}
\newcommand{\tb}[1]{\textcolor{black}{#1}}
\usepackage{mathtools}
\mathtoolsset{showonlyrefs}

\usepackage{braket}

\usepackage{xspace}

\usepackage{float}
\restylefloat{table}
\usepackage{booktabs}
\usepackage[bookmarksnumbered,bookmarksopen,unicode,colorlinks=true,urlcolor=blue,citecolor=blue,linkcolor=blue]{hyperref}

\usepackage{graphicx}
\usepackage{subfig}
\usepackage{stackengine}

\usepackage{enumerate}

\usepackage{bbm}

\usepackage{calc}

\usepackage{accents}

\newtheorem{theorem}{Theorem}[section]
\newtheorem{corollary}[theorem]{Corollary}

\newtheorem{lemma}[theorem]{Lemma}

\newcommand{\R}{\mathbbm{R}}

\newcommand{\qm}[1]{``#1''}

\newcommand{\ie}{i.e.\ }

\newcommand{\cf}{cf.\ }
\newcommand{\eg}{e.g.\ }

\newcommand{\wrt}{w.r.t.\ }

\usepackage[bookmarksnumbered,bookmarksopen,unicode]{hyperref}
\usepackage{url}            

\usepackage[disable]{todonotes}

\DeclarePairedDelimiter{\floor}{\lfloor}{\rfloor}

\title{Enhancing Quantum Algorithms for Quadratic Unconstrained Binary Optimization via Integer Programming}

\author{		
	Friedrich Wagner
	\thanks{Fraunhofer Institute for Integrated Circuits IIS, Erlangen, Germany,\\ Department of Data Science, University of Erlangen-Nuremberg, Germany,
	\texttt{friedrich.wagner@iis.fraunhofer.de}}
	\and
	Jonas Nüßlein
	\thanks{Department of Data Science, University of Erlangen-Nuremberg, Germany}\\
	\and
	Frauke Liers
	\thanks{Department of Data Science, University of Erlangen-Nuremberg, Germany}
}

\begin{document}
	
	\maketitle
	\begin{abstract}
To date, research in quantum computation promises potential for outperforming classical heuristics in combinatorial optimization.
However, when aiming at provable optimality,
one has to rely on classical exact methods like integer programming.
State-of-the-art integer programming algorithms can compute strong relaxation bounds
even for hard instances,
but may have to enumerate a large number of subproblems for determining an optimum solution.
If the potential of quantum computing realizes, it can be
expected that in particular finding high-quality solutions for hard problems can be done fast.
Still, near-future quantum hardware
considerably limits the size of treatable problems.
In this work, we go one step into integrating the potentials of quantum and classical techniques for combinatorial optimization.
We propose a hybrid heuristic for the weighted maximum-cut problem
or, equivalently, for quadratic unconstrained binary optimization.
The heuristic employs a linear programming relaxation,
rendering it well-suited for integration into exact branch-and-cut algorithms.
For large instances, we reduce the problem size
according to a linear relaxation such that the
reduced problem can be handled by quantum machines of limited size.
Moreover, we improve the applicability of QAOA, a parameterized quantum
algorithm, by deriving optimal parameters for
special instances which motivates a parameter estimate for
arbitrary instances.
We present numerous computational results from real quantum
hardware.

\end{abstract}

	\section{Introduction}\label{sec:intro}
Mixed-integer programming looks back
upon a long history of successfully developing methods and algorithms.
Although being NP-hard, algorithmic enhancements have lead to the
result that even difficult instances
can be solved to global optimality, often quickly, by modern
algorithms and implementations,~\tb{\cite{Bixby2000,Bixby2002,Koch2022}.}
Nonetheless, there exist practically relevant instances which
exceed the capabilities of state-of-the-art solvers~\tb{\cite{Gleixner2021}}.
On the other hand, recent progress in quantum computation promises
fast, high-quality heuristics for such problems~\tb{\cite{Moll_2018,Abbas2023}}.
\tb{Important examples for such heuristics are quantum annealing~\cite{Tameem2018,McGeoch2023}
	 and the quantum approximate optimization algorithm (QAOA)~\cite{Farhi_2014,Harrigan_2021}}. 
Although it is unclear whether this promise will realize
and near-future quantum hardware \tb{strongly} limits the size of treatable problems~\tb{\cite{Guerreschi2019,JuengerLobe_2021,lykov2022}},
quantum computation is a highly relevant topic that is currently
studied in many research groups. 

In this work, we take a step towards combining the potential of
quantum and classical computation, with the goal of enhancing the
applicability of quantum algorithms such that larger instances can be solved.
We propose a hybrid, heuristic algorithm for the weighted maximum-cut problem (MaxCut).
The algorithm builds on a well-known linear relaxation of MaxCut.
Solving linear relaxations is the basis of branch-and-cut algorithms for integer programming
since they provide upper bounds on an optimum solution value (we consider here the case of a maximization problem).
\tb{Branch-and-cut divides the solution space into subproblems, called \emph{branching},
	 and calculates upper bounds on the best solution value in each subproblem by solving a linear relaxation.
If this upper bound undercuts the best known lower bound, typically given by the value of the best available solution, the subproblem can be pruned, \ie excluded entirely from the search.
Otherwise, the algorithm continues branching.}
Linear relaxations can be \tb{often} solved quickly, even for very large instances~\tb{\cite{Koch2022}}. 
Thus, the proposed heuristic is particularly well-suited for integration in integer programming methods.
Integrating fast and high-quality heuristics in classical branch-and-cut is a commonly used technique to speed
up the solution process~\tb{\cite{Fischetti2011}}.
Such heuristics provide quickly accessible lower bounds.
If there is a gap between upper and lower bound, the algorithm needs to branch and starts enumerating sub-problems.
The need to enumerate a large number of sub-problems often slows down the solution process significantly.
Branch-and-cut performs particularly well when upper and lower bounds
are strong as branching can then be avoided as much as possible.
Heuristics based on linear relaxation solutions are implemented in all state-of-the-art solvers for integer programming.
Typically, they perform an elaborate rounding procedure to derive feasible integer solutions~\tb{\cite{Achterberg2012}}.
In this work, we solve the so-called \emph{cycle relaxation} of MaxCut,
which has proven valuable for rounding heuristics in previous works~\tb{\cite{Barahona_1988,Liers_2003,Bonato_2014}}.
However, we emphasize that \emph{any} relaxation can readily be substituted in our method.
The proposed heuristics fixes variables based on rounding a cycle relaxation solution.
Together with careful algorithm engineering, this allows to shrink the problem size to a specified target value
such that it can be handled by quantum hardware of limited size.
Thus, instead of relying on branching to determine an integer solution, the enumerative part is heuristically solved by quantum computation.
Finally, having retrieved a solution to the reduced problem from the quantum algorithm,
a solution to the original MaxCut instance is recovered by undoing the shrinking operations appropriately.
This procedure combines the advantages of integer programming, \ie the ability to solve large relaxations,
with the advantages of quantum computation, \ie the ability to determine high-quality solutions quickly.

However, current quantum algorithms for combinatorial optimization are typically parameterized and parameters need to be tuned in a feedback-loop
in order to retrieve satisfactory results.
This heavily increases runtime which would render such quantum heuristics impractical for integration into integer programming solvers.
To tackle this issue, we develop a method to derive good parameters efficiently for QAOA in the case of weighted MaxCut,
speeding up the quantum part of the proposed method \tb{by several orders of magnitude}.
The proposed method benefits from decades of developing linear programming techniques for MaxCut
which we briefly summarize here.

\paragraph*{Linear programming and MaxCut.}
Given a weighted graph, MaxCut
asks for a partition of the nodes such that the weight of connecting edges is maximized.
The decision version of MaxCut is one of Karp's 21 NP-complete problems \cite{Karp_1972}.
Thus, polynomial time algorithms are known for special cases only,
\eg for planar graphs \cite{Hadlock_1975,Liers_2012}, for graphs without long odd cycles \cite{Groetschel_1984} and for graphs with no $K_5$ minor \cite{Barahona_1986}.
The work of Barahona and Mahjoub \cite{Barahona_1986} paved the way
for a long series of successful applications of linear programming (LP) methods for MaxCut.
For a comprehensive book, we refer the interested reader to \cite{Deza_1997}.
A major reason for this success is the development of sophisticated
techniques that allow the solution of
large LP relaxations of MaxCut, \cf \eg \cite{Barahona_1989,Bonato_2014,Juenger_2019,Juenger_2021}.
From a theoretical point of view, these linear relaxations have a rather large worst-case integrality gap of 2 \cite{Avis2003,Vega2007}.
\footnote{Here, the integrality gap of a relaxation is defined as the relaxation optimum value divided by an optimum integer solution value.}
In practice, however, they are typically much stronger, \ie they yield a tight upper bound on the optimum cut value,
making them valuable for branch-and-cut algorithms, \cf \cite{Barahona_1988,Barahona_1989,Liers_2003,Rehfeldt_2022,Charfreitag_2022}.
For LP based methods, it is well-known that the difficulty of solving
MaxCut instances to global optimality strongly depends on the
density of the underlying graph.
Indeed, dense instances in the order
of 100 vertices can already go beyond what is solvable in practice by
modern algorithms. However, if instances are sparse and possibly
defined on regular structures such as grid graphs, or if the weights
are specifically designed, state-of-the-art methods can handle problems with over 10,000
vertices in less than a minute, \cf \cite{Charfreitag_2022}.

Another widespread relaxation of MaxCut is based on semi-definite programming (SDP).
It was first introduced by Goemans and Williamson \cite{Goemans1995}
who developed a famous $\alpha$-approximation algorithm for MaxCut with $\alpha\approx0.878$.
Indeed, it was shown later that $\alpha$ is the best possible approximation factor for a polynomial-time algorithm if the Unique Games Conjecture holds \cite{Khot2007}.
Moreover, the integrality gap of the SDP-relaxation is precisely $1/\alpha\approx1.14$ \cite{Feige2002}, which is considerably smaller than the integrality gap of LP relaxations.
In practice, SDP based exact algorithms perform well even on dense instances, contrary to LP-methods.
However, they scale (much) worse than LP-methods \wrt to the number of vertices, \cf \cite{Rendl_2007,Rendl2010,Palagi2012}.
Noteworthy, strong relaxations of MaxCut have proven useful even for problems with additional constraints \cite{Buchheim_2010}.

\paragraph*{Quantum computation and MaxCut.}
Significant technological progress in quantum computation (QC) hardware has been achieved in the last decade,
improving both digital (\cf \eg \cite{Arute_2019,Hyypp_2022}) and analog quantum computers (\cf \eg \cite{Hauke_2020,Dwave_2021}).
The availability of quantum hardware has driven researchers to seek for practical quantum advantage in various fields of applications.
Among those, combinatorial optimization plays a prominent role (\cf \eg \cite{Abbas2023}), although clear advantages have not yet been shown \tb{\cite{Ronnow2014,JuengerLobe_2021}}.
Here, commonly used algorithms are the QAOA in digital QC (\cf \cite{Farhi_2014}),
and quantum annealing in analog QC (\cf \cite{Hauke_2020}).
Although from a theoretical point of view, both algorithms are applicable to a broad range of combinatorial optimization problems,
research is mainly focused on quadratic unconstrained binary optimization (QUBO).
\tb{QUBO and MaxCut are equivalent problems as there
	exists a linear transformation between them.
	In particular, any QUBO problem on $n$ variables can be transformed into an equivalent MaxCut instance on $n+1$ vertices (\cf \cite{Hammer_1965, Barahona_1989,DESIMONE_1990}).
	We detail this transformation in~Sec.~\ref{sec:prelim}.}
The reason for the prominence of QUBO in QC is that many current quantum hardware platforms, both digital and analog, have a natural connection to QUBO.
\footnote{In the physics literature, usually the term \qm{Ising model} is used instead of QUBO.}
Therefore, the usual approach is to model the problem of interest as a QUBO problem \cite{Lucas_2014}.
Moreover, technical restrictions of current hardware limit the size of treatable QUBO problems.
The maximum treatable problem size depends \tb{not only on the hardware size but also heavily} on the problem density. 
\tb{Here, analog QC encounters significant overhead in quantum resources when embedding dense problem graphs onto sparse hardware graphs~\cite{Choi2008,Choi2010}.
Similarly, in digital QC, qubit routing adds additional quantum overhead for dense problems~\cite{WBLW_2023}.
In digital QC, also error correction~\cite{nielsen_chuang_2010} or error mitigation~\cite{van_den_Berg_2022,Temme_2017} require additional resources.}
It is reasonable to expect that, even if QC achieves clear experimental advantages,
such restrictions due to quantum hardware persist.
Therefore, algorithmic size-reduction techniques will be necessary to leverage the full potential of QC.

\paragraph*{Our contribution.}
In this work, we present a novel hybrid quantum-classical heuristic for MaxCut which
can easily be integrated into modern integer programming solvers.
The proposed algorithm overcomes the limitations of near-future quantum hardware
by employing classical size-reduction techniques.
\tb{In the context of classical algorithms, our main contribution is the combination of graph shrinking with linear programming.}
This allows to combine the advantages of both quantum and classical methods for combinatorial optimization:
The former \tb{can be expected to} quickly produce high-quality solutions for hard instances \tb{(\cf \eg~\cite{Bravyi_2020,Farhi_2022,Basso2022})} but the size of treatable problems is limited.
The latter can efficiently compute strong relaxations, \tb{in particular on sparse} instances of large size~\cite{Charfreitag_2022}.

Additionally, we improve the applicability of the quantum algorithm used in our experiments -- QAOA for weighted MaxCut --
by deriving optimal parameters for triangle-free, regular graphs with weights following a binary distribution.
\tb{To this end, we derive an alternate formula for the expectation value of depth-1 QAOA applied to weighted MaxCut by extending the work of~\cite{Wang_2018}.
	This formula is less expensive to evaluate numerically than the formula derived in Refs.~\cite{Ozaeta_2022, Bravyi_2020} on sparse instances.}
Furthermore, this leads a new estimate for good
parameter values for arbitrary instances.
As a result, the runtime of the quantum algorithm is reduced,
which is especially beneficial when integrating it into classical branch-and-cut.

We present various experimental results from quantum simulators and real quantum hardware, showing
the applicability of the proposed method on example instances
inspired by spin glass physics. 
The quantum processor in our experiments has 27 qubits of which we use at most 10 due to strong noise.
This limits the size of the reduced problem to a maximum of 10 vertices.
Due to the hybrid algorithm that first reduces the instance size \tb{classically},
we solve instances with up to 100 vertices within 60~s of total computation time, 
\tb{including classical runtime and runtime on real quantum hardware}.
For \tb{standard QAOA without classical size reduction}, these instances exceed
the currently available resources.
We emphasize, that all instances considered in this work can be solved in reasonable time
by purely classical methods.
Owing to \tb{the limited size of the quantum hardware used in this work}, larger instances would not give deeper insights
in the quantum part since the algorithm outcome would mainly be determined by the classical part.
Our experiments thus yield a proof-of-principle,
encouraging the application of the proposed method when quantum hardware advances.

\paragraph*{Structure.} The remainder of this paper is organized as follows:
\tb{First, we give an overview of and discuss connections to related literature in Section~\ref{sec:rel}.}
In Section~\ref{sec:prelim}, we formally introduce the MaxCut and QUBO problem as well as QAOA.
Section~\ref{sec:algo} describes our algorithm in detail.
In Section~\ref{sec:param}, we derive the novel result on optimum QAOA-parameters.
Section~\ref{sec:exp} presents various experimental results.
We conclude with a summary and indicate further directions of research in Section~\ref{sec:concl}.

\section{Related work}\label{sec:rel}
In the context of QAOA, a similar size-reduction procedure, \tb{known as \emph{recursive QAOA} (RQAOA)}, has been proposed in \cite{Bravyi_2020} \tb{and was further developed in~\cite{Bravyi_2022,finzgar2023}}.
In those works, however, the motivation is quite different.
The authors reduce large-scale problems by iteratively solving them approximately by QC
in order to reach regimes where classical exact methods can be applied.
While their size-reduction technique is similar to ours, they use QC to determine the next reduction step.
Thus, for practical applicability, large-scale quantum hardware would be required.
\tb{This is also the case for the algorithm proposed in~\cite{Dupont_2024} which employs variable correlations generated by QAOA in order to obtain valid QUBO solutions via a rounding procedure.}

The authors of \cite{Tate_2023} \tb{and \cite{Egger2021}} combine classical semi-definite relaxations of MaxCut with QAOA.
They employ relaxation solutions for warm-starting QAOA in order to achieve better results.
On the contrary, in this work, relaxation solutions are used to reduce problem size.
Furthermore, in the context of QA, size-reduction was successfully applied in \cite{DwaveHybrid_2023}.
There, the problem size is reduced -- similarly to this work -- by fixing a set of variables.
The authors repetitively choose a random set of variables to be fixed and try to improve an overall solution by solving the reduced problem via QA.
This is in contrast to this work, where the set of variables to be fixed is determined by a relaxation solution.

The performance of QAOA for MaxCut was studied by the authors of~\cite{lykov2022} \tb{and~\cite{Guerreschi2019}}.
Their results indicate a high threshold for QAOA to achieve an advantage over classical heuristics.
However, their study is restricted to unweighted 3-regular graphs, whereas our work deals with arbitrary, weighted graphs.

\tb{The authors of \cite{Ozaeta_2022} derive an analytical expressions for the expectation value produced by depth-1 QAOA on general Ising models, which includes weighted MaxCut.
	The same analytical expression for weighted MaxCut was already obtained in~\cite{Bravyi_2020}.
	In this work, we derive a different formula for weighted MaxCut by extending a result on unweighted MaxCut from~\cite{Wang_2018} and~\cite{hadfield2018quantum}.
	The resulting, alternative formula is cheaper to evaluate
    than the formula from~\cite{Ozaeta_2022}
    and~\cite{Bravyi_2020} for sparse instances.}

\tb{Moreover, the authors of~\cite{Ozaeta_2022} derive parameters which are optimal on regular graphs in the ensemble average for normally distributed weights.
Similarly, Ref.~\cite{Farhi_2022} gives an explicit formula for the ensemble average of the expectation value for arbitrary-depth QAOA applied to the Sherrington-Kirkpatrick model of infinite size.
Therein, the authors derive optimal parameters for the depth-1 case.
In contrast, in this work we derive (approximate) optimal parameters for a given instance rather than averaging over an instance ensemble.}
Furthermore, deducing high-quality QAOA-parameters for weighted MaxCut was studied in \cite{Shaydulin_2023}.
There, the authors successfully transfer parameters which are known to be well-suited for unweighted MaxCut to the weighted case.
On the contrary, the method for QAOA-parameter estimation proposed in this work is purely based on characteristics of the weighted MaxCut instance at hand.

Fixing variables in order to reduce the problem size is an established pre-processing technique in integer programming algorithms.
Here, one seeks for variable assignments that are provably part of an optimum solution since then solving the reduced problem will recover an optimum solution to the original instance.
For MaxCut, such methods are developed in \eg \cite{Lange_2018,Ferizovic2020,Rehfeldt_2022}.
Clearly, an advantage of such techniques is that it allows to reduce the problem size without loosing the possibility to retrieve an optimum solution.
On the other hand, they are computational\tb{ly} more expensive than rounding relaxation solutions and reaching a specific target size of the reduced problem cannot be ensured.
This, however, is crucial in the context of this work since quantum hardware strictly limits the size of treatable problems.

Considering heuristics for integer programming, rounding variables according to a relaxation solution is a commonly used paradigm.
In the context of MaxCut, it was already applied in the early work of Barahona et al.~\cite{Barahona_1988}
and, more recently, \eg in~\cite{Liers_2003,Bonato_2014}.
In those works, a set of variables corresponding to a spanning tree is rounded, as this uniquely defines a cut.
The spanning tree is chosen such that its variables are as close to integer as possible.
Computationally, this can be implemented efficiently by a minimum-weight spanning-tree algorithm.
In this work, however, only a subset of variables (which corresponds to a forest) is fixed by rounding,
leaving the remaining variables free to be optimized by the quantum algorithm.
Furthermore, shrinking was successfully applied to other closely
related combinatorial optimization problems in~\cite{Groetschel_1984,Liers_2011}.

	\section{Preliminaries}\label{sec:prelim}
\tb{Before stating our results, we introduce the prerequisites necessary for the upcoming sections.
First, we formally define the MaxCut problem and introduce its integer programming model.
Then, we define the QUBO problem and recap the transformation between QUBO and MaxCut.
Finally, we introduce QAOA and its application to MaxCut.}

\paragraph*{MaxCut and QUBO} As stated earlier, MaxCut and QUBO 
are well known to be equivalent problems.
For ease of presentation, we introduce the maximum cut problem first as this is
natural for integer programming methods.
Given an undirected graph $G=(V,E)$ and a vertex subset $W\subseteq V$,
the edge set $\delta(W) \coloneqq \{uv \in E\mid u\in W, v\not \in W\}$ is called a \emph{cut} of $G$.
For edge weights $w \in\R^{|E|}$, the weight of a cut $\delta(W)$ is defined as $\sum_{e \in \delta(W)}w_e$.
The MaxCut problem asks for a cut of maximum weight.
\begin{figure}
	\centering
	\subfloat{\includegraphics[height=2.5cm]{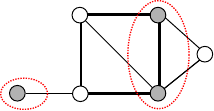}\label{fig:cut}}\hspace{2cm}
	\caption{\tb{Example for a cut. A vertex subset $W$ is marked in gray. All edges crossed by the red dashed lines are in the cut $\delta(W)$. An exemplary cycle of length four is marked by thick edges.
			 It intersects with the red dashed line two times. In general, any cycle intersects with $\delta(W)$ an even number of times}.
	}
	\label{fig:cut}
\end{figure}
An edge subset $C = \set{v_0v_1,v_1v_2,\dots,\tb{v_{k-1}v_k}, v_kv_0} \subseteq E$ is called a \emph{cycle}.
Clearly, a cut and a cycle \tb{always} coincide in an even number of edges as illustrated in Fig.~\ref{fig:cut}.
Algebraically, this observation can be modeled by the so-called \emph{odd-cycle inequalities}~\tb{\cite{Barahona_1986}}.
If $C$ is a cycle and $x\in \set{0,1}^{|E|}$ is the edge incidence-vector of a cut, it holds
\begin{align}
	\sum_{e \in Q}x_e-\sum_{e \in C\setminus Q}x_e\leq |Q|-1 \quad\forall Q \subseteq C,\ |Q|\ \mathrm{odd}\,.
\end{align}
In fact, the odd-cycle inequalities for all cycles $C$ are sufficient
to define a cut.
Thus, a widely used integer linear programming formulation (\cf \eg \cite{Barahona_1988,Barahona_1989,Juenger_2019,Rehfeldt_2022}) of MaxCut is
\begin{subequations}
	\begin{align} 
		\max_{\tb{x}}\ &\sum_{e \in E}w_e x_e\label{eq:Objective}\\
		\text{s.t. }&\sum_{e \in Q}x_e-\sum_{e \in C\setminus Q}x_e\leq |Q|-1 \quad\forall Q \subseteq C,\ |Q|\ \mathrm{odd},\ \forall C \subseteq E\ \mathrm{cycle} \label{eq:oddcycle}\\
		& 0\leq x_e\leq 1 \quad\forall e \in E \label{eq:bounds}\\
		& x_e \in \{0,1\} \quad\forall e \in E \label{eq:bin}\ .
	\end{align}
\end{subequations}
The \emph{cut polytope}, first introduced in \cite{Barahona_1986}, is the convex hull of all cut incidence-vectors,
\begin{align}
	P_{\mathrm{CUT}} \coloneqq \mathrm{conv}\{x \in \R^{|E|}|\eqref{eq:oddcycle}-\eqref{eq:bin}\}\ .
\end{align}
Dropping the integrality condition \eqref{eq:bin},
the model \eqref{eq:Objective}-\eqref{eq:bounds} is called the \emph{cycle relaxation} of MaxCut.
In general, a solution to the cycle relaxation yields an upper
bound on the optimum cut value.
However, it is known that the cycle relaxation has integer optimal solutions for all weights $w_e\in \R^{|E|}$ if and only if $G$ has no $K_5$ minor, \cf \cite{Barahona_1986}.
Although the MaxCut problem is NP-hard,
optimizing the cycle relaxation can be done in polynomial time via \emph{odd-cycle separation}.
First, only \eqref{eq:Objective} and \eqref{eq:bounds} is optimized.
Given an optimal solution $x^*$, the odd-cycle separation algorithms decides whether
$x^*$ satisfies all odd-cycle inequalities.
If not, it returns a violated one which is added to the model and the procedure is repeated.
Otherwise, $x^*$ is optimal for the cycle relaxation.
Barahona and Majhoub \cite{Barahona_1989} give a polynomial-time algorithm for odd-cycle separation 
based on shortest paths.
Typically, one seeks for odd-cycle inequalities that belong to \emph{chordless}
cycles as they define facets of the cut
polytope, \cf \cite{Barahona_1986,Juenger_2019}.
For complete graphs, the odd-cycle inequalities take the form
\begin{subequations}
	\begin{align} 
		x_{tu}-x_{uv} -x_{vt} 	&\leq 0  \label{eq:tr1}\\
		-x_{tu}+x_{uv}-x_{vt} 	&\leq 0  \label{eq:tr2}\\
		-x_{tu}-x_{uv}+x_{vt} 	&\leq 0  \label{eq:tr3}\\
		x_{tu} +  x_{uv} + x_{vt}&\leq 2  \label{eq:tr4}
	\end{align}
\end{subequations}
\noeqref{eq:tr2,eq:tr3}
for all triangles $\set{t,u,v}$ in $G$.
A key ingredient of our algorithm is, that a solution to the cycle relaxation
can be computed efficiently. 
We thus use an optimum cycle relaxation solution to reduce the size of the
MaxCut instance such that it can be handled by near-term quantum computers.
\tb{To this end, we identify an edge whose corresponding value of the relaxation solution is close to integer.
Then, the two incident vertices are replaced by
a single vertex.
We refer to such a new vertex originating from replacing two vertices as a \emph{super-vertex}.
We repeat this vertex replacement until a target size is reached.}
This process is called
\emph{shrinking}, \cf~\cite{Liers_2011, Bonato_2014}.
Details are explained in Section~\ref{sec:algo}.

Current research in quantum computation for combinatorial optimization mainly focuses on QUBO problems,
that is, problems of the form
\begin{subequations}
	\begin{align} 
		\max_{\tb{x}\ }&\sum_{i,j}q_{ij}x_ix_j\label{eq:ObjectiveQ}\\ 
		\text{s.t.\ }&x_i \in \{0,1\}\quad \forall i \in \{1,\dots,n\}\ \label{eq:binQ}\,,
	\end{align}
\end{subequations}
where $q_{ij} \in \R$.
A QUBO problem with $n$ variables can be transformed to an equivalent MaxCut problem on $n+1$ vertices, \cf~\cite{Hammer_1965,Barahona_1989,DESIMONE_1990,Juenger_2021}.
\tb{
To this end, we consider the complete graph $K_{n+1}$ with vertices $\{0,1,\dots,n\}$.
For an edge $ ij\in K_{n+1} $ with $i,j > 0$, we define its weight as $w_{ij} = q_{ij} + q_{ji}$.
Moreover, for an edge $i0$ with $i > 0$, we set $w_{i0}=\sum_{j=1}^n  q_{ij} + q_{ji}$.
Then, a cut $\delta(W)$ with weight $W$ gives rise to a QUBO solution with value
$ Q = -W/2 + C $ where $C = 1/4\left(\sum_{e\in E}w_e + 2\sum_i q_{ii} + \sum_{i<j}q_{ij}+q_{ji} \right)$.
The QUBO solution is obtained by setting $x_i = 0$ if $i0$ is in the cut $\delta(W)$ and $x_i = 1$ otherwise.
}

Our size-reduction method works on MaxCut, but we keep in mind that
this MaxCut problem might originate from a transformed QUBO problem.
After size-reduction, we formulate the reduced MaxCut problem as a QUBO problem in order to use QC for its solution.
The natural QUBO formulation of MaxCut is
\begin{subequations}
	\begin{align} 
		\max_{\tb{x}}\ &C(x)=\sum_{uv \in E}w_{uv}(x_{u}+x_{v}-2x_{u}x_{v})\label{eq:Objectivex}\\ 
		\text{s.t. }&x_v \in \{0,1\}\quad \forall v \in {V}\ \label{eq:binx} .
	\end{align}
\end{subequations}
This model is then solved by QAOA and a solution to the original MaxCut (or QUBO) problem is reconstructed.
\noeqref{eq:binx}

\paragraph*{QAOA.}
QAOA is a quantum-classical hybrid algorithm, originally proposed by Fahri et al. in \cite{Farhi_2014}.
Since then, it has received great attention, which led to the development of more sophisticated versions and variants, see \eg \cite{Bravyi_2020,Hadfield_2019}.
QAOA computes approximate solutions of arbitrary, unconstrained, binary
optimization problems defined by a \tb{target} function $C:\{0,1\}^{n} \rightarrow \R$.
The goal is to find an $x^{*}\in\{0,1\}^{n}$ maximizing the \tb{target} function.
QAOA is a parameterized algorithm with real-valued parameters $\boldsymbol{\gamma} =  (\gamma_{1}, ..., \gamma_{p})$ and $\boldsymbol{\beta} = (\beta_{1}, ..., \beta_{p})$.
The hyper-parameter $p$, called \emph{depth}, controls the \tb{computational} complexity of the algorithm.
QAOA prepares the quantum state
\begin{align}
	\ket{\psi(\boldsymbol{\beta},\boldsymbol{\gamma})}=e^{-i\beta_{p}H_{M}}e^{-i\gamma_{p}H_{C}}...e^{-i\beta_{1}H_{M}}e^{-i\gamma_{1}H_{C}} \ket{+}\ .
\end{align}
Here, $H_{C}$ is the problem-specific \emph{phase Hamiltonian}, defined by
\begin{align}\label{eq:phaseh}
	H_{C}\ket{x} = C(x)\ket{x}~\forall x \in \{0,1\}^{n}
\end{align}
and $H_{M}$ is the problem-independent \textit{mixing Hamiltonian}, defined by
\begin{align}
	H_{M}=\sum_{i=1}^{n}X_{i}\ ,
\end{align}
where $X_{i}$ is the Pauli $X$-operator acting on qubit $i$.
The initial state $\ket{+}$ is the uniform superposition over all basis states, that is
\begin{align}
	\ket{+} = \frac{1}{\sqrt{2^{n}}}\sum_{x \in \{0,1\}^{n}}\ket{x}\ .
\end{align}
Performing a single measurement yields a bit string $x$ with probability $P_{\boldsymbol{\gamma},\boldsymbol{\beta}}(x)=|\braket{x|\psi(\boldsymbol{\beta},\boldsymbol{\gamma})}|^2$.
Thus, the expectation value of $C(x)$ evaluates to
\begin{align}\label{eq:qaoa_exp}
	F(\boldsymbol{\beta},\boldsymbol{\gamma}) \coloneqq  \sum_{x \in \{0,1\}^n}P_{\boldsymbol{\gamma},\boldsymbol{\beta}}(x)C(x) = \langle\psi(\boldsymbol{\beta},\boldsymbol{\gamma})\lvert H_{C}\rvert\psi(\boldsymbol{\beta},\boldsymbol{\gamma})\rangle\ .
\end{align}
The distribution $P_{\boldsymbol{\gamma},\boldsymbol{\beta}}(x)$ is parameter-dependent,
and one seeks for parameters that yield high probabilities for high value solutions.
Usually, this is done by classically maximizing the expectation value $	F(\boldsymbol{\gamma},\boldsymbol{\beta})$
which is estimated by an average over a finite sample
\begin{align}\label{eq:qaoa_avg}
	\langle C \rangle (\boldsymbol{\gamma},\boldsymbol{\beta})= \frac{1}{N}\sum_{i=1}^{N} C(x_i)\ ,
\end{align}
where $N$ is the total number of samples and $x_i\in \set{0,1}^n$ is the $i$-th sample.
We remark, that another, less common metric for parameter optimization
is the probability of sampling an optimal solution.
\tb{If the value of an optimum solution is unknown, as is typically the case for MaxCut,}
this metric is only applicable for small instances because it requires \tb{calculating the optimum solution value} beforehand.

In the experiments, we use QAOA for the weighted-MaxCut target-function given in \eqref{eq:Objectivex}.
It is easily verified, that the phase Hamiltonian
\begin{align}\label{eq:maxcuth}
	H_{C} = \sum_{uv \in E}\frac{w_{uv}}{2}(I-Z_{u}Z_{v})
\end{align}
fulfils \eqref{eq:phaseh} for this particular \tb{target} function.
Here, $Z_{u}$ is the Pauli $Z$-operator acting on qubit $u$.
Having introduced the necessary prerequisites, we now describe the proposed quantum-classical algorithm in more detail.

	\section{Algorithm description}\label{sec:algo}

In this section, we describe the proposed hybrid algorithm for MaxCut problems.
It can be divided into four major steps.
First, \emph{correlations} between vertex pairs are computed from an
optimum cycle relaxation solution.
Here, the closeness of a relaxation variable to an integer value
is interpreted as a tendency of the corresponding vertex pair
to lie in equal or opposite partitions in an optimum cut.
Second, the problem size is reduced by imposing correlations,
that is, vertex pairs with large absolute correlations are identified.
Third, the shrunk problem is solved by QAOA.
Finally, a feasible solution to the original problem is reconstructed
by undoing the shrinking operations appropriately.

\paragraph*{Computing Correlations.}
To reduce problem size of an instance too large to be solved by
current quantum hardware, the algorithm relies on correlations.
A correlation between a vertex pair quantifies the tendency of the vertices
pair being in equal or opposite partitions in an optimum cut.
More formally, for a subset $S\subseteq V\times V$ of vertex pairs, correlations are
a set $\set{b_{uv}|uv \in S}$ where $b_{uv}\in [-1,1]$.
Correlations are called \emph{optimal}, if there is an optimum cut $\delta(W)$
such that $b_{uv}\geq0$ ($b_{uv}<0$) if $u$ and $v$ lie in equal (opposite) partitions in $\delta(W)$.
In general, the closeness of $b_{uv}$ to $1$ ($-1$) is interpreted as the tendency of $u$ and $v$ lying in equal (opposite) partitions.

In principle, any method to deduce correlations can be used in the algorithm.
However, we compute correlations from a solution $x^*$ to the cycle relaxation by
\begin{align}\label{eq:cort}
	b_{uv} \coloneqq 1 - 2x^*_{uv}\in [-1,1]\ .
\end{align}
It is well known and can also be seen in our numerical experiments,
that cycle relaxation solutions indeed often resemble correlations from an optimum integer solution.

\paragraph*{Shrinking.}
\begin{figure}
	\centering
	\subfloat[]{\includegraphics[height=2.5cm]{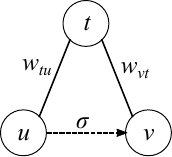}\label{fig:skizze_vertex_indent_vor}}\hspace{2cm}
	\subfloat[]{\includegraphics[height=2.5cm]{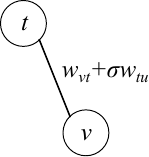}\label{fig:skizze_vertex_indent_nach}}
	\caption{Sketch for vertex shrinking.
		(a): Vertices $u$ and $v$ are to be identified, where
		$\sigma\in \{-1,1\}$ defines whether $u$ and $v$ lie in the equal or opposite partitions.
		Vertex $t$ is a neighbor of $u$.
		(b): In the shrunk MaxCut instance, vertex $v$ is a super-vertex containing $u$ and $v$ with adjusted edge weights.
		In the case where edge $vt$ is not present in (a), it is constructed as shown in (b) with $w_{vt}=0$.
	}
	\label{fig:vertex_ident}
\end{figure}
We reduce problem size by identifying vertex pairs which have a large absolute correlation.
This process is illustrated in Fig.~\ref{fig:vertex_ident}.
First, we describe the process of shrinking a single pair of vertices.
To this end, let $b_{uv}$ be a correlation and define
\begin{align}
	\sigma \coloneqq \begin{cases}
		\mathrm{sign}(b_{uv})&\ b_{uv} \neq 0\\
		1&\ b_{uv} = 0\,.
	\end{cases} 
\end{align}
If $\sigma = 1$ ($\sigma =-1$), we enforce $u$ and $v$ to lie in equal (opposite) partitions.
Solving MaxCut on $G=(V,E)$ with this additional constraint is equivalent to 
solving MaxCut on a graph $G'=(V',E')$, where $V' = V\setminus\{u\}$.
For this reduction, edge weights need to be adjusted.
For $u\in V$, denote by $\mathcal{N}(u)\coloneqq\set{v\in V|uv \in E}$ the neighborhood of $u$.
For all $t \in \mathcal{N}(u)$, define new weights by
\begin{align}\label{eq:nw_def}
	w_{vt}' \coloneqq \begin{cases}
		\sigma w_{tu} &\text{ if }vt \not\in E \\
		w_{vt} + \sigma w_{tu} &\text{ if }vt \in E\ ,
	\end{cases}
\end{align}
compare Fig.~\ref{fig:vertex_ident}.
All other edge weights remain unchanged.
Vertex $v$ now represents a super-vertex containing vertices $u$ and $v$.
Multiple edges are replaced by a single edge.
Thus, the reduced MaxCut instance is defined on $G'=(V',E')$ with
\begin{align}\label{eq:e_def}
E' = \Big(E \cup \{vt:t \in \mathcal{N}(u)\}\Big)\setminus \Big\{ut:t \in \mathcal{N}(u)\Big\}\ ,
\end{align}
and $V' = V\setminus \{u\}$.
Any cut in $G'$ can be translated to a cut in $G$ if $\sigma$ is known.

This shrinking process is iterated until a target problem-size is reached.
We shrink in descending order of the absolute values $|b_{uv}|$.
If two vertices to be shrunk have already been identified to the same super-vertex in a previous iteration,
the shrinking step is skipped.
This avoids possibly contradictory variable fixings.
From a different perspective, we only fix variables corresponding to a forest, \ie a cycle-free subgraph.
For the mapping from a solution of the shrunk problem back to a solution of the original instance,
it is necessary to keep track of the vertex identifications.
In this work, we use QAOA for solving the shrunk problem.
However, any suitable method for this task -- quantum or classical -- can in principle be substituted in our algorithm.

The overall algorithm has three desirable properties which can easily be verified.
First, it returns an optimum solution if shrinking with optimal correlations and if the shrunk problem is solved to optimality.
Second, if shrinking is performed with non-optimal correlations but the shrunk problem is solved to optimality,
the quality of the returned solution will increase when shrinking less vertices.
Third, for a fixed shrunk problem, the quality of the returned solution increases with the solution quality of the shrunk problem.
Having described the classical algorithmic framework,
we now turn to the quantum part.

	\section{QAOA-Parameter Estimate for Weighted MaxCut}\label{sec:param}
We apply depth-1 QAOA to solve the shrunk MaxCut problem.
\tb{We restrict our study to depth $p=1$ for two reasons.
	First, the high noise level in current quantum hardware
    prohibits a meaningful execution of large-depth QAOA in practice.
	This is also the case for the state-of-the-art quantum processor used in our experiments in Section~\ref{sec:exp}.
	Second, for depth-1 QAOA we are able to derive a high-quality parameter estimate which renders parameter optimization via a classical feedback-loop unnecessary,
	thus simplifying the execution of QAOA on quantum hardware significantly.
	However, we emphasize that a successful execution of
    larger-depth QAOA would lead to improved solutions.
	Indeed, our experiments in Section~\ref{sec:exp} reveal that
    depth-1 QAOA often falls short in returning optimum
    solutions. Thus, it is desirable to successfully implement
    large-depth QAOA in the future when quantum hardware improves.}

The solution quality returned by QAOA heavily depends on the parameters $\gamma,\beta$.
Therefore, deriving good parameters is crucial for its success.
The authors of \cite{Wang_2018} give an analytical expression for the expectation value $F(\gamma,\beta)$
as defined in \eqref{eq:qaoa_exp} for depth-1-QAOA applied to unweighted MaxCut.
In this section, we extend their results to efficiently derive good
parameters for depth-1-QAOA, when applied to weighted MaxCut.

The following statements might serve as a starting point for a further classical parameter optimization.
For all instances considered in this work, however, the parameter
estimate performs well enough to be used without any further
parameter optimization.
Next, we state the extended result for weighted MaxCut.
\begin{lemma}\label{lem:exp}
	Let $G=(V,E)$ graph with edge weights $w\in \R^{|E|}$.
	Let $\gamma, \beta \in \R$ and $F(\gamma,\beta)$ be defined as in \eqref{eq:qaoa_exp} with $H_C$ given in \eqref{eq:maxcuth}.
	Further, for $u,v \in V$, let $\mathcal{N}_u(v)$ be the set of neighbours of $v$ excluding $u$ and denote by $\Lambda(u,v)$ the set of common neighbours of $u$ and $v$.
	Then it holds
	\begin{align}
		F(\gamma,\beta)=\sum_{uv\in E}f_{uv}(\gamma,\beta), \label{eq:expthF}
	\end{align}
	where
	\begin{align}
		\begin{aligned}
			f_{uv}(\gamma,\beta) &= w_{uv} \Bigg[ \frac{1}{2} + \frac{1}{4}\sin(4\beta)\sin(\gamma w_{uv})\left(\prod_{s\in \mathcal{N}_v(u)}\cos(\gamma w_{us}) + \prod_{t \in \mathcal{N}_u(v)}\cos(\gamma w_{vt}) \right) \\
			&- \frac{1}{2}\sin^2(2\beta)\sum_{\substack{ N\subseteq\Lambda(u,v)\\ |N| = 1,3,5,\dots}}\prod_{s\in \mathcal{N}_v(u)\setminus N}\cos(\gamma w_{us}) \prod_{t \in \mathcal{N}_u(v)\setminus N} \cos(\gamma w_{vt}) \prod_{r \in N}\sin(\gamma w_{ur})\sin(\gamma w_{vr}) \Bigg]\ . \label{eq:expth}
		\end{aligned}
	\end{align}
\end{lemma}
\begin{proof}
The proof follows the concept from \cite{Wang_2018} \tb{and~\cite{hadfield2018quantum}}.
We need to evaluate the expectation value of the costs associated to a specific edge $uv\in E$,
\begin{align}\label{eq:f}
	f_{uv} &= \bra{+} e^{i\gamma H_{C}}e^{i \beta H_M} \frac{1}{2}(1-Z_uZ_v) e^{-i \beta H_M}e^{-i\gamma H_{C}} \ket{+}\\
	&= \frac{1}{2}\big(1 - \bra{+} e^{i\gamma H_{C}}e^{i \beta H_M} Z_uZ_v e^{-i \beta H_M}e^{-i\gamma H_{C}} \ket{+}\big)\ ,
\end{align}
where $H_{C}$ is the MaxCut Hamiltonian defined in \eqref{eq:maxcuth}.
For the conjugation with $e^{i\gamma H_{C}}$, it is convenient to neglect the constant terms.
Thus, we define
\begin{align}
	\tilde{H}_{C}\coloneqq - \frac{1}{2} Z_u Z_v + \tilde{H}_{C_u} + \tilde{H}_{C_v}\,
\end{align}
with
\begin{align}
	\tilde{H}_{C_u}\coloneqq - \frac{1}{2} \sum_{s\in \mathcal{N}_v(u)} Z_u Z_s.
\end{align}
Now, with $c=\cos 2\beta$ and $s=\sin 2\beta$, it holds
\begin{align}\label{eq:conjb}
	e^{i \beta H_M} Z_uZ_v e^{-i \beta H_M} 
	= c^2 Z_uZ_v + sc (Y_uZ_v + Z_uY_v) + s^2 Y_uY_v\ .
\end{align}
Each term is conjugated with $e^{i\gamma H_{C}}$ separately.
The key observation is that only terms proportional to products of $X$-operators will contribute to the expectation value
since $\bra{+}Y\ket{+} = \bra{+}Z\ket{+} = 0$.
Thus, the first term in \eqref{eq:conjb} does not contribute.
We now turn to the second term.
Let $c_{ij}' = \cos (w_{ij}\gamma)$ and $s'_{ij} = \sin (w_{ij}\gamma$). By using $YZ = - ZY$ we have
\begin{align}
	e^{i\gamma H_{C}} Y_{u}Z_v e^{-i\gamma H_{C}} &= e^{i\gamma Z_u Z_v} e^{i2\gamma \tilde{H}_{C_u}} Y_{u}Z_v \\
	&= (Ic'_{uv}- i s'_{uv} Z_u Z_v)\prod_{s\in \mathcal{N}_v(u)} (Ic'_{us} -is'_{us} Z_u Z_s) Y_u Z_v\ .
\end{align}
When expanding the product, there is only a single term that will contribute to the expectation value.
This is the term which results from choosing only factors $Ic'_{us}$ in the product. This term is proportional to
\begin{align}
	Z_uZ_vY_uZ_v = -iX_u\ .
\end{align}
Thus
\begin{align}
	\bra{+} e^{i\gamma H_{C}} Y_{u}Z_v e^{-i\gamma H_{C}} \ket{+} = -s'_{uv} \prod_{s\in \mathcal{N}_v(u)}c'_{us}
\end{align}
By symmetry, it follows for the third term in \eqref{eq:conjb},
\begin{align}
	\bra{+} e^{i\gamma H_{C}} Z_{u}Y_v e^{-i\gamma H_{C}} \ket{+} = -s'_{uv} \prod_{t\in \mathcal{N}_u(v)}c'_{tv}\,.
\end{align}
Now, we turn to the last term in \eqref{eq:conjb}. With $YZ = - ZY$, it follows
\begin{align}
	e^{i\gamma H_{C}} Y_{u}Y_v e^{-i\gamma H_{C}} &=
	e^{2i\gamma \tilde{H}_{C_u}} e^{2i\gamma \tilde{H}_{C_v}} Y_{u}Y_v  \\
	&=   \prod_{s\in \mathcal{N}_v(u)} (c'_{us} I - i s'_{us} Z_uZ_{s})\prod_{t\in \mathcal{N}_u(v)} (c'_{tv}I - i s'_{tv} Z_{t}Z_v)  Y_{u}Y_v   \,.
\end{align}
When expanding the product, only terms proportional to $X_uX_v$ will contribute to the expectation value.
If we choose a single term $Z_uZ_{s}$ in the first product, and a single term  $Z_tZ_{v}$ in the second product such that $s = t$ (\ie a common neighbor of $u$ and $v$),
the resulting term will be proportional to $X_uX_v$ since
\begin{align}
	Z_u Z_s Z_t Z_v Y_u Y_v = Z_u Z_v Y_u Y_v = -X_u X_v
\end{align}
This is also the case for any odd combination of common neighbors.
In fact, those are exactly the terms which are proportional to $X_u X_v$.
Thus, summation over all odd combinations of common neighbors gives
\begin{align}
	\bra{+} e^{i\gamma H_{C}} Y_{u}Y_v e^{-i\gamma H_{C}} \ket{+} 
	= 	\sum_{\substack{ N\subseteq\Lambda(u,v)\\ |N| = 1,3,5,\dots}}\prod_{s\in \mathcal{N}_v(u)\setminus N}c'_{us} \prod_{t \in \mathcal{N}_u(v)\setminus N} c_{vt} \prod_{r \in N}s_{ur}s_{vr}\ .
\end{align}
Finally, substituting \eqref{eq:conjb} in \eqref{eq:f} yields the result stated in \eqref{eq:expth}.

\end{proof}
\tb{The formula~\eqref{eq:expth} in Lemma~\ref{lem:exp} is cheap to evaluate numerically on sparse instances.
	For a given edge $uv\in E$ ,the number of terms in the sum in Eq.~\eqref{eq:expth} is in $O(2^{|\Lambda(u,v)|}$).
	For sparse instances, $|\Lambda(u,v)|$ is small, often zero, such that the sum in \eqref{eq:expth} has only view terms.
	In contrast, the computational cost of the formula derived in~\cite{Ozaeta_2022} and~\cite{Bravyi_2020} does not depend on $|\Lambda(u,v)|$ but grows linearly with $|N(uv)|$, which is larger than $O(2^{|\Lambda(u,v)|}$)	on sparse instances.
	On the other hand, the number of terms in~\eqref{eq:expth} is very large for dense instances.
	Then, the formula in~\cite{Ozaeta_2022} and~\cite{Bravyi_2020} can be evaluated faster.}

An often studied class of instances consists of weights that are
chosen following a binary distribution in $\{-a,a\},\ a>0$. Then, for specific graph topologies, maximizers of $F(\gamma,\beta)$ can be
derived analytically from Lemma~\ref{lem:exp}, as we show next.

\begin{corollary}\label{cor}
	For triangle-free, $d$-regular graphs with weights taking values in $\{-a,a\},\ a>0$, \eqref{eq:expthF} is maximized for
	\begin{align}
		\gamma = \frac{1}{a} \arctan\left(\frac{1}{\sqrt{d-1}}\right)\ ,\quad\beta=\frac{\pi}{8}\ .\label{eq:optparams}
	\end{align}
\end{corollary}
\begin{proof}
	\tb{The proof is conducted analogous to~\cite{Wang_2018}, Corollary 1, and~\cite{hadfield2018quantum}, Corollary~2.}
	Using $\Lambda(u,v)=0$, $|\mathcal{N}_u(v)|=|\mathcal{N}_u(v)|=d-1$, $w_{uv}\in \{-a,a\}$, \eqref{eq:expth} simplifies to
	\begin{align}\label{eq:expc} 
		f_{uv}(\gamma,\beta)&=w_{uv} \Bigg[ \frac{1}{2} + \frac{1}{2}\sin(4\beta)\sin(\gamma w_{uv})\cos^{d-1}(\gamma a) \Bigg]\\ 
		&=\frac{w_{uv}}{2} + \frac{a}{2}\sin(4\beta)\sin(\gamma a)\cos^{d-1}(\gamma a)\ .
	\end{align}
	By differentiation \wrt $\beta$ and $\gamma$, it is easily verified that \eqref{eq:optparams} maximizes \eqref{eq:expc}.
	Since the maximizers \eqref{eq:optparams} do not depend on $uv$, they also maximize \eqref{eq:expthF}.
\end{proof}

Although possibly not always being the best choice, Corollary~\ref{cor} motivates the following parameter guess for arbitrary weighted graphs:
\begin{align}\label{eq:params}
	\bar{\gamma} = \frac{1}{\bar{a}} \arctan\left(\frac{1}{\sqrt{\bar{d}-1}}\right)\ ,\quad\bar{\beta}=\frac{\pi}{8}\ .
\end{align}
Here, $\bar{a}$ is the mean of absolute weight values,
\begin{align}
	\bar{a} \coloneqq \frac{1}{|E|}\sum_{uv \in E} |w_{uv}|\ ,
\end{align} 
and $\bar{d}$ is the average vertex degree.
For triangle-free, regular graphs with weights in $\{-a,a\}$, \eqref{eq:params} reduces to \eqref{eq:optparams}.

As mentioned above, our numerical experiment show a good performance of parameters \eqref{eq:params} such that we use them without further optimization.
\tb{Of course, better parameters might exist and could possibly be found by a classical parameter optimization.
However, this amounts to solving a non-convex, continuous optimization problem
 which requires multiple estimations of $F(\gamma,\beta)$ via sampling from quantum hardware, thus increasing the runtime of QAOA significantly
compared to our approach, which keeps parameters fixed.}

Having introduced the algorithmic framework,
we are now ready to discuss experimental results from our implementations.

	\section{Experimental Results}\label{sec:exp}

In this section, we present various computational results for the proposed algorithm.
All implementations are done in Python.
We use the open-source quantum-software development-kit \emph{Qiskit} \cite{Qiskit_2021} for quantum circuit construction, ideal quantum simulation as well as quantum hardware communication.
For graph operations we use the package \emph{networkx} \cite{hagbger2008networkx}.
Integer models and relaxations are solved via the Python interface of the solver \emph{Gurobi} \cite{Gurobi}.
Quantum hardware experiments are performed on the quantum backend \emph{ibmq\_ehningen} \cite{ibmq}
which has 27 superconducting qubits.
However, in our experiments we use at most 10 qubits 
since the influence of noise
becomes prohibitively large for higher qubit numbers.
 
\subsection{QAOA Parameter Prediction via (\ref{eq:params})}\label{sec:exp_params}
First, we evaluate the performance of the quantum part of our algorithm, which is a depth-1-QAOA with predetermined parameters as stated in \eqref{eq:params}.
\tb{As mentioned earlier, the high noise level in the quantum hardware used in our experiments prohibits the use of larger-depth QAOA,
	 which would improve the obtained solutions.
Additionally, for depth-1-QAOA we can exploit the parameter estimate~\eqref{eq:params} which renders parameter optimization unnecessary.}

We compare the quality of the estimated parameters to the best possible parameter choice.
As a performance metric, we measure the average value of the produced cut size on different weighted MaxCut instances.
For each instance, we simulate the quantum algorithm on an ideal, noise-free device, \ie we numerically evaluate $F(\gamma,\beta)$ via \eqref{eq:expth}.
Additionally, we perform experiments on quantum hardware.
Here, we evaluate \eqref{eq:qaoa_avg}
with a sample size of $N=1024$ and $C$ as defined in \eqref{eq:Objectivex}.
In order to find \tb{close-to-optimal} parameters, we perform an exhaustive grid search.
Parameters $(\gamma,\beta)$ are chosen from a grid with step size $0.1$ on $[0,\pi/2]\times[0,\pi/2]$.
Limiting the parameter search space is eligible because of symmetry relations in QAOA, \cf \cite{Zhou_2020}.
\tb{In principle, a smaller step size would yield better parameters.
In our experiments, however, we observe that a step size of $0.1$ is sufficient to capture the characteristics of the parameter-landscape, \cf also Figs.~\ref{fig:pred} and \ref{fig:pred3}}.

We consider three sets of test instances.
The first set is constructed solely to investigate the quality of the predetermined parameter choice \eqref{eq:params}.
We construct instances fulfilling all or only some of the assumptions
under which these parameters are provably optimal for an ideal quantum device.
Recall from Corollary~\ref{cor}, that these assumptions are:
\begin{enumerate}[(i)]
	\item The graph is regular.
	\item The graph is triangle-free.
	\item The weights are chosen from the set $\{-a,a\}$ for some $a>0$.
\end{enumerate}
Owing to the limitations of quantum hardware,
we restrict the size of instances to the minimum which
still allows for resembling the desired characteristics.
Thus, all instances from the first set have only four vertices and can be solved by \tb{hand}.
They are depicted in Fig.~\ref{fig:param_instances}.
\begin{figure}[]
	\centering
	\subfloat[]{\includegraphics[width=0.25\linewidth]{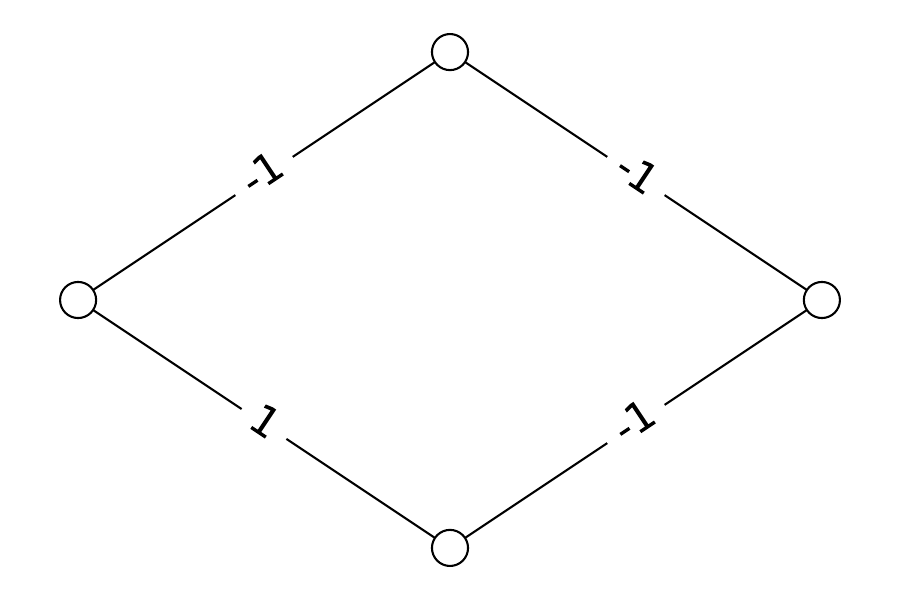}}
	\subfloat[]{\includegraphics[width=0.25\linewidth]{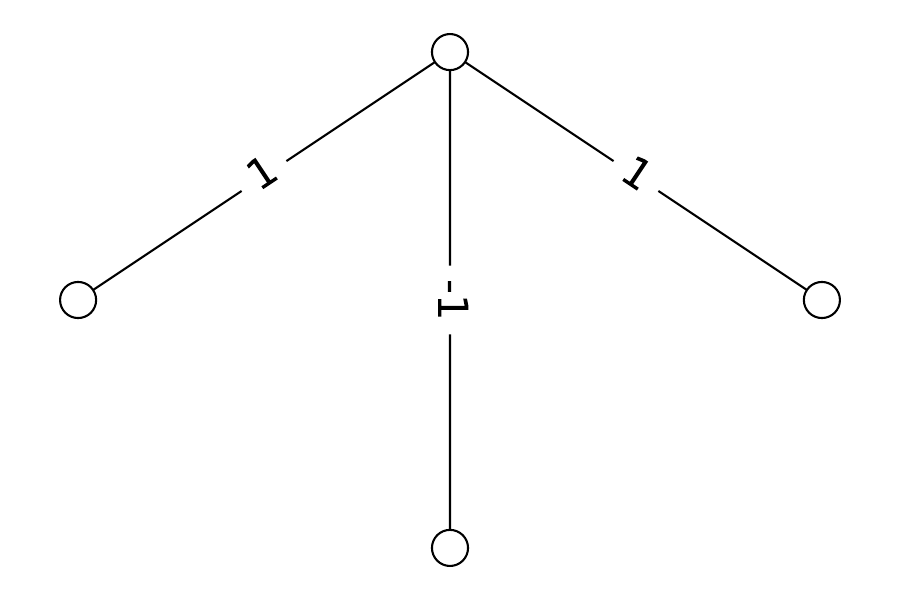}}
	\subfloat[]{\includegraphics[width=0.25\linewidth]{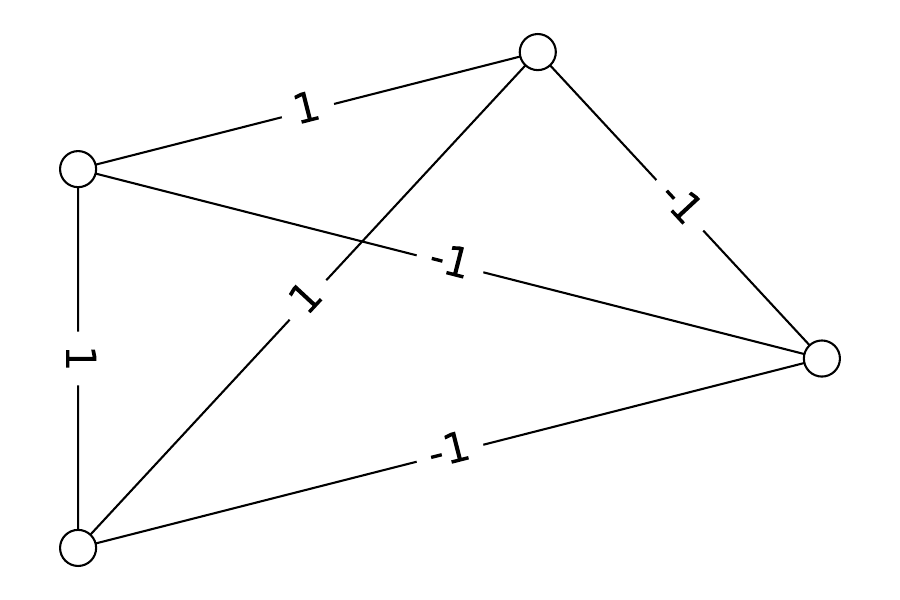}}\\
	\subfloat[]{\includegraphics[width=0.25\linewidth]{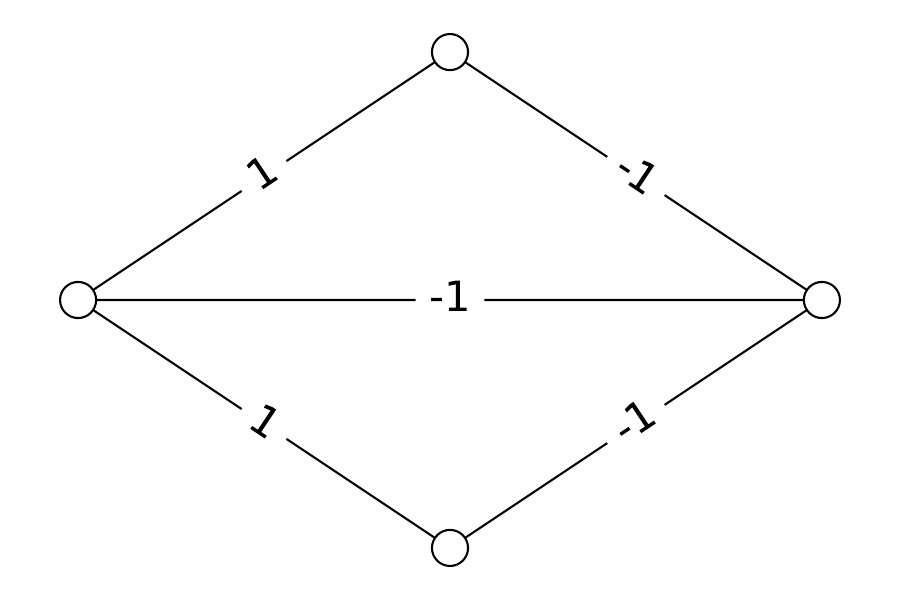}}
	\subfloat[]{\includegraphics[width=0.25\linewidth]{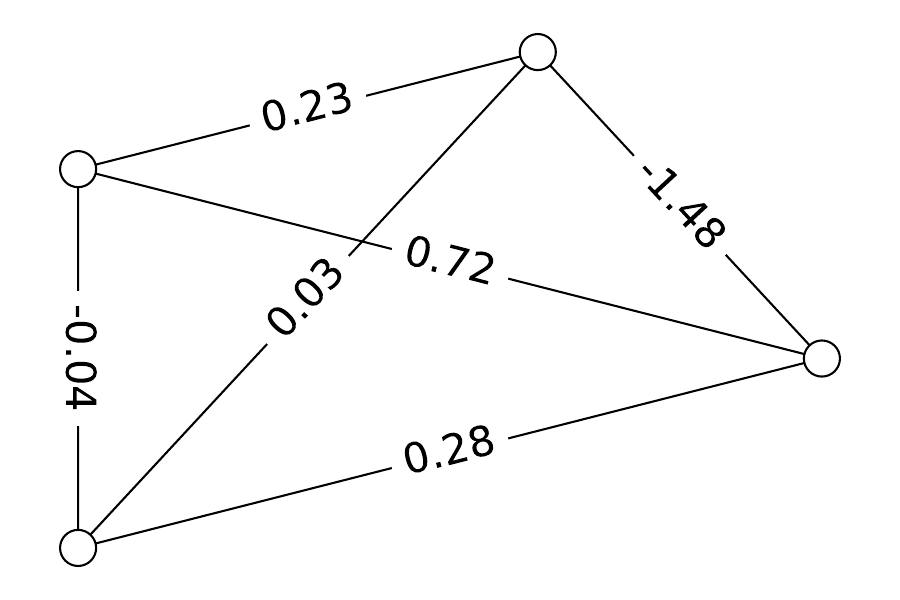}}
	\subfloat[]{\includegraphics[width=0.25\linewidth]{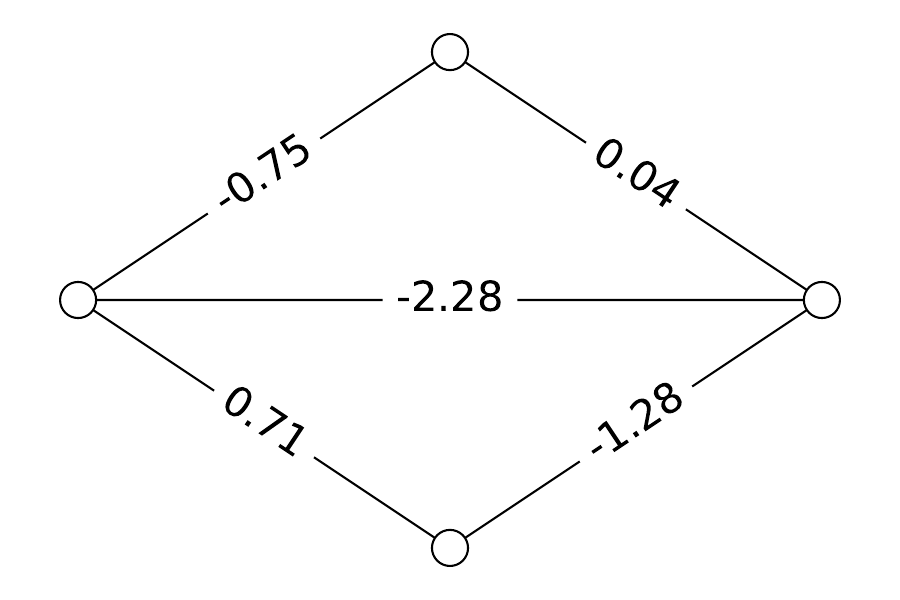}}
	\caption{Test instances for the parameter-estimate evaluation.}
	\label{fig:param_instances}
\end{figure}
As topologies, we choose a ring, a star, a complete graph and a ring with chord.
Weights are drawn either uniformly at random from $\pm 1$ or from a normal distribution.
These weight distributions are motivated from spin glass physics \cf \eg \cite{Liers_2003}.

The second set contains MaxCut instances which result from shrinking
a $2\times3$ grid with $\pm1$ weights: an instance considered in the next section,
where we evaluate the combination of shrinking and QAOA.
Instance names of the second set are of the form \qm{2x3g$ns$}.
Here, $n$ denotes the number of vertices in the shrunk graph
and $s$ marks whether the cycle relaxation was used for shrinking (\qm{c})
or random edges were shrunk (\qm{r}), see Section~\ref{sec:exp_shrink} for details.
If random shrinking and cycle-relaxation-shrinking led to the same
instance, \qm{rc} is used.

The third set is constructed to investigate the scalability of our parameter estimate.
Here, we consider a 10-regular graph on 100 vertices (\qm{100reg}) as well as a 100-vertex ring
with 20 additional edges inserted such that the graph remains triangle-free (\qm{100trf}).
\tb{These two instances have normal distributed weights.
Moreover, we consider four 100-vertex Erdős–Rényi random graph ensembles with densities $d\in \set{0.05, 0.1, 0.15, 0.2}$ (\qm{100rand$d$}).
These graphs are unweighted, i.e., edges have unit weights.
It is known that random unweighted MaxCut undergoes a phase transition at the critical edge density of $d^*=1/|V|$, \cf~\cite{Coppersmith_2003}.
If $d<d^*$, the expected number of edges not in an optimum cut is $\Theta(1)$ for large $|V|$.
If $d>d^*$, the expected number of edges not in an optimum cut jumps to $\Theta(|V|)$.
Thus, we consider four Erdős–Rényi ensembles: a sub-critical ensemble with $d=0.05$, a critical ensemble with $d=d^*=0.1$ and two super-critical ensembles with $d\in \set{0.15,0.2}$.
Here, we average over 20 graphs from each ensemble.}

Table~\ref{tab:param_pred} summarizes the experimental results.
Instance names of the first set correspond to sub-figures in Fig.~\ref{fig:param_instances}.
For each instance, we mark which of assumptions (i)-(iii) are met.
Furthermore, we measure the relative deviation of $F(\bar{\gamma},\bar{\beta})$ from the maximum of $F(\gamma,\beta)$,
\begin{align}
	\frac{\max F(\gamma,\beta)-F(\bar{\gamma},\bar{\beta})}{\max F(\gamma,\beta)-\min F(\gamma,\beta)}\in [0,1]\ .
\end{align}
Thus, a value of 0 means that indeed the maximum is hit, whereas 1 means that we hit the minimum instead.
\tb{Here, we estimate the true maxima and minima via the respective values from the grid search.}
Analogously, for the real quantum experiment, we measure
\begin{align}
	\frac{\max \langle C \rangle(\gamma,\beta)-\langle C \rangle(\bar{\gamma},\bar{\beta})}{\max \langle C \rangle(\gamma,\beta)-\min \langle C \rangle(\gamma,\beta)}\in [0,1]\ .
\end{align}
As expected, the deviation of $F(\bar{\gamma},\bar{\beta})$ from the optimum value of $F(\gamma,\beta)$ increases when the instance violates more assumptions from (i)-(iii).
The maximum observed deviation is 10\% on instance f which violates all assumptions.
From an integer-programming point of view, 10\% might seem a large deviation.
However, we stress that we do not compare single solution values but expectation values as QAOA is a probabilistic algorithm.
Of course, the best solution in a reasonably sized sample will be better than the expectation value.

In Table~\ref{tab:param_pred}, the values of $\langle C \rangle (\bar{\gamma},\bar{\beta})$, measured with the real quantum backend, roughly follow the qualitative behavior of $F(\bar{\gamma},\bar{\beta})$ across all instances.
Quantitatively, the values from real hardware always lie slightly above the corresponding values from the ideal simulation.
This means, that our parameter estimate performs slightly worse on real hardware than in theory.
Of course, in real quantum hardware, many physical effects influence the outcome of QAOA which are all not considered in the derivation of \eqref{eq:params}.
With this in mind, it is even more encouraging that our parameter
estimate hits the optimum within a maximum observed deviation of
at most 13\%.

Considering the instances \qm{2x3g$ns$}, resulting from shrinking the $2\times3$ grid,
observed deviations are typically less than for the artificial instances a - f.
This further motivates the use of QAOA with parameters \eqref{eq:params},
when combined with shrinking as investigated in the next section.

The large instances with 100 vertices could not be executed on real hardware (\qm{-} in Table~\ref{tab:param_pred}).
\tb{Still, we can evaluate~\eqref{eq:expth} numerically.}
We observe a good performance of the parameter estimate in the numerical simulations \tb{for all large instances with a maximum deviation of 1 \%.}
\tb{Although the estimate seems to become slightly worse with increasing density, }
these results indicate a high scalability of the proposed method.

\begin{table}[t]
	\makebox[\textwidth][c]{
	\begin{tabular}{cccccc}
		\toprule
		Instance & Regular & Triangle-free & $w_{ij}\in \set{-a,a} $ &$\frac{\max F(\gamma,\beta)-F(\bar{\gamma},\bar{\beta})}{\max F(\gamma,\beta)-\min F(\gamma,\beta)}$ & $\frac{\max \langle C \rangle(\gamma,\beta)-\langle C \rangle(\bar{\gamma},\bar{\beta})}{\max \langle C \rangle(\gamma,\beta)-\min \langle C \rangle(\gamma,\beta)}$\\
		\midrule
		a  & \cmark & \cmark & \cmark & 0.0 & 0.03   \\
		b  & \xmark & \cmark & \cmark & 0.0 & 0.02   \\
		c  & \cmark & \xmark & \cmark & 0.07 & 0.02   \\
		d  & \xmark & \xmark & \cmark & 0.05 & 0.09   \\
		e  & \cmark & \xmark & \xmark & 0.08 & 0.08   \\
		f  & \xmark & \xmark & \xmark & 0.10 & 0.13   \\
		2x3g2r  & \cmark & \cmark & \cmark & 0.0 & 0.0   \\
		2x3g3rc  & \cmark & \xmark & \cmark & 0.06 & 0.06   \\
		2x3g4c  & \cmark & \cmark & \xmark & 0.0 & 0.0   \\
		2x3g4r  & \xmark & \cmark & \cmark & 0.0 & 0.01   \\
		2x3g5c  & \xmark & \xmark & \xmark & 0.0 & 0.04   \\
		2x3g5c  & \xmark & \xmark & \cmark & 0.0 & 0.05   \\
		2x3g6rc  & \xmark & \xmark & \cmark & 0.0 & 0.0   \\
		100reg  & \cmark & \xmark & \xmark & 0.007 & -   \\
		100trf  & \xmark & \cmark & \xmark & 0.02 & -   \\
		\tb{100rand0.05}  & \xmark & \xmark & \cmark & 0.0 & -   \\
		\tb{100rand0.1}  & \xmark & \xmark & \cmark & 0.0 & -   \\
		\tb{100rand0.15}  & \xmark & \xmark & \cmark & 0.007 & -   \\
		\tb{100rand0.2}  & \xmark & \xmark & \cmark & 0.01 & -   \\
		\bottomrule
	\end{tabular}
	\caption{Experimental results for evaluating the QAOA-parameter-estimate. In column \qm{Instance}, a - f refer to Fig.~\ref{fig:param_instances}.
	In columns \qm{Regular}, \qm{Triangle-free} and \qm{$w_{ij}\in \set{-a,a}$} fulfilled assumptions of Cor.~\ref{cor} are marked. The second last column gives the relative deviation of $F(\bar{\gamma},\bar{\beta})$ from the  optimal value. The values are calculated by \eqref{eq:expth}. The last column gives the deviation of $\langle C \rangle (\bar{\gamma},\bar{\beta})$ from the optimum. Here, real quantum hardware was used.}
	\label{tab:param_pred}
	}
\end{table}
To further illustrate the differences between ideal simulation and real experiment,
$F(\gamma,\beta)$ and $\langle C \rangle (\gamma,\beta)$ are visualized in Fig.~\ref{fig:pred}.
Although not being exactly identical, the data from the experiment qualitatively follows the simulation.
It is worth noting, that the absolute values in the experiment are usually smaller
than in the simulation.
The observed differences between simulation and experiment are due to noise effects in quantum hardware.
Analogue figures for other instances appear in Appendix \ref{app:exp_param}.
\begin{figure}[]
	\centering
	\subfloat[Instance a.]{\includegraphics[width=0.49\linewidth]{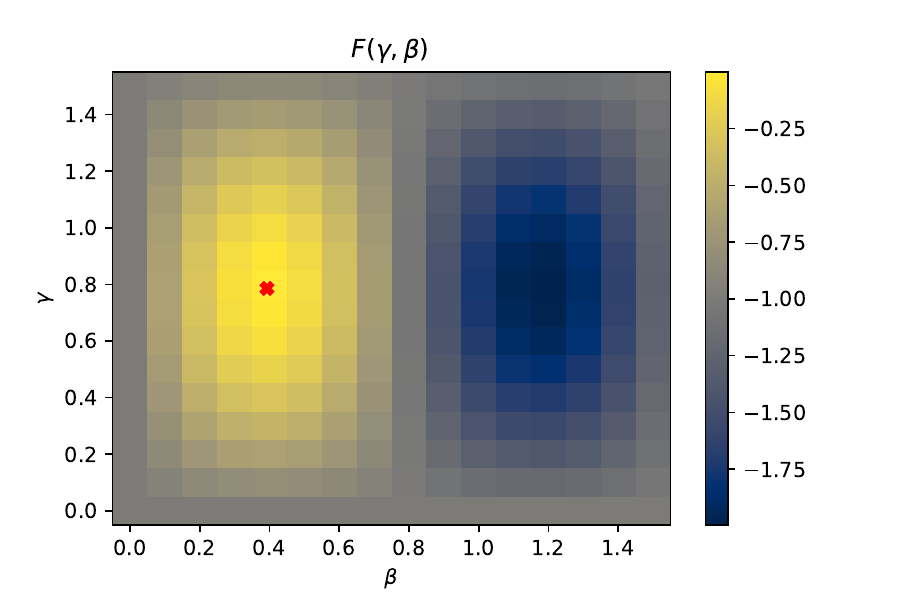}\includegraphics[width=0.49\linewidth]{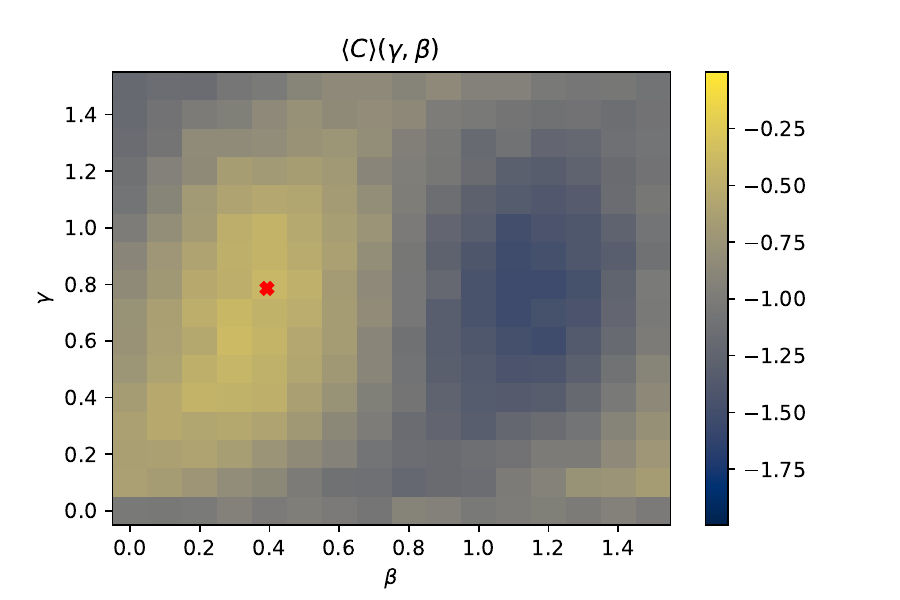}\label{fig:pred_ring}}\\
	\subfloat[Instance b.]{\includegraphics[width=0.49\linewidth]{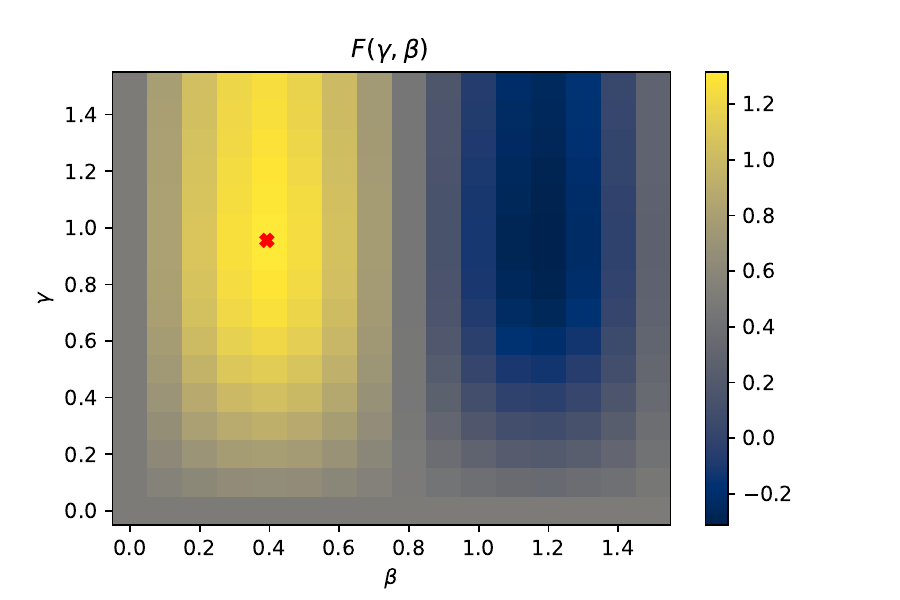}\includegraphics[width=0.49\linewidth]{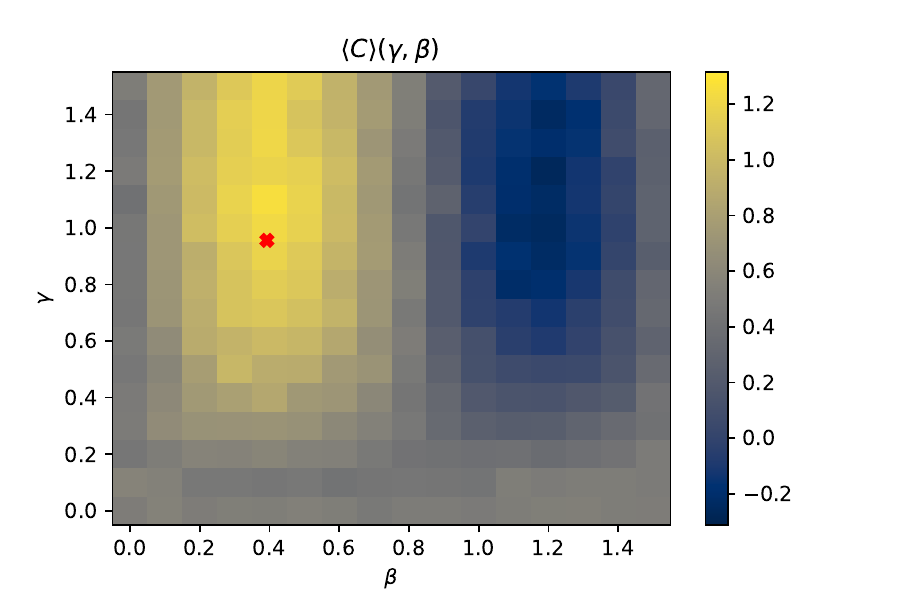}\label{fig:pred_star}}\\
	\subfloat[Instance c.]{\includegraphics[width=0.49\linewidth]{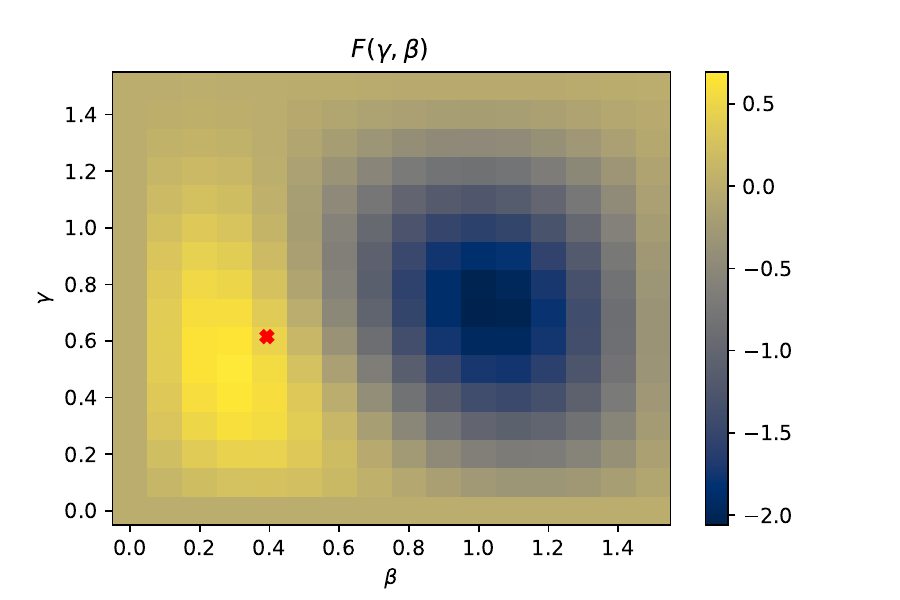}\includegraphics[width=0.49\linewidth]{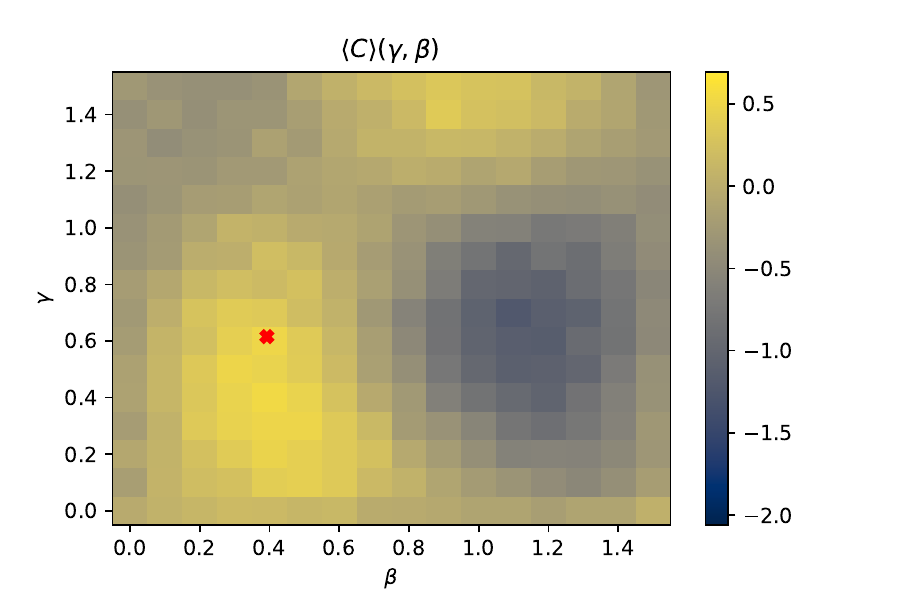}\label{fig:pred_compl}}
	\caption{Visualized results for QAOA parameter-estimate on instances a, b and c from Fig.~\ref{fig:param_instances}. The x- and y-axis represent values of the parameters $\beta$ and $\gamma$, respectively. The red cross marks the estimate in \eqref{eq:params}. The color encodes the expectation value (left) or the average (right) of the cut size. On the left, we mimic an ideal quantum device by evaluation of \eqref{eq:expth}. On the right, values are results from the quantum hardware. Here, every pixel represents the average taken over $1,024$ samples.}
	\label{fig:pred}
\end{figure}

Summarizing, our results show that the parameter estimate \eqref{eq:params}
performs well, even on instances that do not satisfy the assumptions
where it is provably optimal.
These results encourage us to use of QAOA without further parameter optimization
in the upcoming section where we combine shrinking with
QAOA to solve MaxCut instances too large to be handled by quantum hardware alone.

\subsection{Combining Shrinking with QAOA}\label{sec:exp_shrink}
In this section, we combine shrinking via the cycle relaxation, as described in Section~\ref{sec:algo}, with QAOA.
Again, instance sizes are kept small due to limitations of quantum hardware.
As in the previous section, graph topologies and weight distributions are motivated by spin glass physics.
We emphasize that all instances can be solved by classical
integer programming quickly.
Thus, the experiments in this section should be considered as a proof-of-principle rather than a performance-benchmark.

When computing the cycle relaxation on sparse graphs,
we model a complete graph and assign zero weights to edges not present in the sparse graph.
Thus, all odd-cycle inequalities belong to triangles and are of the form \eqref{eq:tr1}-\eqref{eq:tr4}.
We remark, that this relaxation has the same objective value as the sparse odd-cycle relaxation \eqref{eq:Objective}-\eqref{eq:bounds}
since the polyhedron defined by \eqref{eq:oddcycle}-\eqref{eq:bounds}
is a projection of the polyhedron \eqref{eq:tr1}-\eqref{eq:tr4}
along a direction orthogonal to the cost vector in \eqref{eq:Objective}.
We work with the dense cycle-relaxation for two reasons.
First, the dense formulation has variables for all vertex pairs, not only for edges present in the sparse graph.
This allows to shrink vertex pairs not connected by an edge in the sparse graph, \ie vertex pairs corresponding to an edge with zero weight in the dense model.
Second, this formulation can be implemented straightforwardly by simply enumerating all triangles.
Of course, this implementation is not efficient compared to state-of-the-art separation methods.
However, since we aim at a proof-of-concept and considered instances are small, 
this approach is sufficient for our numerical experiments.

Table~\ref{tab:shrink} summarizes the instance data and results.
\begin{table}
	\makebox[\textwidth][c]{
		\begin{tabular}{ccrrcrrrrrr}
			\toprule
			     Instance      &          Graph          & $|V|$ & $|E|$ &         Weights         & $r_{0}$ & \tb{$r_1$} & \tb{$r_{\floor*{|V|/2}}$} & \tb{$r_{|V|-6}$} & $r_{|V|-4}$ & \tb{$r_{|V|-2}$} \\ \midrule
			       2x3b        &     $2\times3$ grid     &     6 &     7 &         $\pm1$          &   53 \% & \tb{66 \%} &                \tb{98 \%} & \tb{53 \%}       &       90 \% &      \tb{100 \%} \\
			       k6b         &          $K_6$          &     6 &    15 &         $\pm1$          &   47 \% & \tb{68 \%} &                \tb{86 \%} & \tb{ 47 \%}      &       73 \% &      \tb{ 94 \%} \\
			       k6n         &          $K_6$          &     6 &    15 &      $\mathcal{N}$      &   44 \% & \tb{38 \%} &                \tb{79 \%} & \tb{ 44 \%}      &       56 \% &      \tb{ 94 \%} \\
			       3x3b        &     $3\times3$ grid     &     9 &    12 &         $\pm1$          &   57 \% & \tb{71 \%} &                \tb{90 \%} & \tb{ 83 \%}      &       90 \% &      \tb{ 96 \%} \\
			       k10b        &        $K_{10}$         &    10 &    45 &         $\pm1$          &   69 \% & \tb{71 \%} &                \tb{82 \%} & \tb{ 76 \%}      &       74 \% &      \tb{ 88 \%} \\
			       k10n        &        $K_{10}$         &    10 &    45 &      $\mathcal{N}$      &   51 \% & \tb{55 \%} &                \tb{79 \%} & \tb{ 73 \%}      &       74 \% &      \tb{ 98 \%} \\
			       4x4b        &     $4\times4$ grid     &    16 &    24 &         $\pm1$          &       - &     \tb{-} &                \tb{73 \%} & \tb{ 80 \%}      &       86 \% &      \tb{ 98 \%} \\
			   t2g10\_5555     & $10\times10$ torus grid &   100 &   200 &      $\mathcal{N}$      &       - &     \tb{-} &                  \tb{-  } & \tb{ 96 \%}      &       98 \% &      \tb{ 99 \%} \\
			ising2.5-100\_5555 &        $K_{100}$        &   100 &  4950 & $\mathcal{N}$, decaying &       - &     \tb{-} &                  \tb{-  } & \tb{ 98 \%}      &       99 \% &      \tb{ 99 \%} \\ \bottomrule
		\end{tabular}
		\caption{Instance data and results for evaluating the shrinking algorithm. $|V|$ and $|E|$ are the number of vertices and edges, respectively. In column \qm{Weights}, $\pm1$ abbreviates the uniform distribution on $\{-1,1\}$, $\mathcal{N}$ abbreviates the standard normal distribution and \qm{$\mathcal{N}$,~decaying} is a normal distribution decaying proportional to euclidean distance.
		In columns \tb{$r_{k}$, we state the achieved approximation ratio when shrinking \tb{$k$}} vertices \tb{and solving the shrunk problem by depth-1-QAOA on actual quantum hardware}.
		\tb{A \qm{-} indicates that the problem instance is too large for the quantum hardware used in our experiments.}
	}
		\label{tab:shrink}
	}
\end{table}
Concerning the runtime for the small instances with up to 16 vertices, solving the cycle relaxation took far less than one second while the total quantum runtime for 10,000 executions of QAOA was roughly 2 seconds.
We also considered two larger instances from literature, available in \cite{Biqmac_2007}.
The first is \qm{t2g10\_5555}, a $10\times 10$ toroidal grid with 100 vertices
and normal distributed weights.
The second is \qm{ising2.5-100\_5555}, the fully connected graph
$K_{100}$ on 100 vertices	
with weights decaying exponentially with distance.
Both instances are an order of magnitude larger than
the capabilities of current digital QC hardware.
However, they can still be handled by
classical branch-and-cut based on cycle relaxations.
\footnote{Although the instance \qm{ising2.5-100\_5555} is fully connected,
the cycle relaxation is still strong
due to the specific structured weight distribution
which suppresses weights of distant vertices.}
For both large instances, solving the cycle relaxation took less than one minute,
while the runtime \tb{for solving the shrunk instance} on the quantum machine  was roughly 5 seconds.

In Table~\ref{tab:shrink}, we compare the achieved approximation ratio using the real quantum computer without shrinking, $r_{|V|}$, to the approximation ratio when shrinking to four vertices, $r_4$.
In this work, we define the approximation ratio as
\begin{align}\label{eq:appratio}
	r\coloneqq \frac{\langle C \rangle-C_{\mathrm{min}}}{C_{\mathrm{max}}-C_{\mathrm{min}}} \in [0,1] \,,
\end{align}
where $\langle C \rangle$ is the average size of a returned cut, $C_{\mathrm{min}}$ and $C_{\mathrm{max}}$ are the minimum and maximum cut sizes, respectively.
For all of the small instances, we observe that shrinking significantly increases the approximation ratio.
For the larger instances \tb{with 16 and 100 vertices}, running QAOA was not possible without shrinking \tb{below 11 vertices} due to hardware limitations (\qm{-} in Table~\ref{tab:shrink}).
Notably, for the first instance, shrinking only two vertices yields an increase in approximation ratio of almost 40\ \%.

Fig.~\ref{fig:shrink_size} allows a more detailed analysis of the algorithm for the first two instances from Table~\ref{tab:shrink}.
Here, we run the algorithm with different settings.
First, we alter the number of shrunk vertices (x-axis in Fig.~\ref{fig:shrink_size}).
In general, the potential, \ie the best possible cut value, will degrade when shrinking more vertices
since there might not exist an optimum solution with the imposed
correlations.
This is the case if and only if the imposed correlations are not optimal.
In this case, even when the shrunk problem is solved to optimality, the recreated solution cannot be optimal.
On the other hand, shrinking more vertices reduces the problem size, which might lead to better solutions for the shrunk problem,
especially when using near-future quantum hardware.
The second setting we vary are the procedures for computing correlations (corresponding to different colors in Fig.~\ref{fig:shrink_size}).
This allows to investigate the influence of the correlation quality on the overall performance.
Two methods are used:
\begin{enumerate}
	\item Correlations inferred from the cycle relaxation, defined in \eqref{eq:cort} (blue lines in Fig.~\ref{fig:shrink_size}).
	\item All correlations are zero which results in shrinking random vertex pairs with $\sigma=1$ (red lines in Fig.~\ref{fig:shrink_size}). 
\end{enumerate}
Lastly, we apply different solution methods (corresponding to different line markes in Fig.~\ref{fig:shrink_size}) for the shrunk problem
in order to investigate the influence of the sub-problem solution quality on the overall performance.
The sub-problem is solved in four different ways:
\begin{enumerate}
	\item We solve the shrunk problem to optimality by integer programming (solid lines).
	This yields an upper bound on the performance of our algorithm.
	\item We solve the shrunk problem by a depth-1-QAOA with parameters predetermined by \eqref{eq:params}, executed $10,000$ times on an ideal quantum simulator (dashed-dotted lines).
	\item We solve the problem exactly as in 2, but with real
          quantum hardware instead of an ideal quantum simulator (dashed lines).
	\item We solve the shrunk problem randomly by flipping a coin for each vertex,
		\ie we assign each vertex to either partition with probability $1/2$ (dotted lines).
		When noise in the real quantum machine becomes large (called \emph{decoherence}),
		the quantum computer effectively performs the coin-flipping heuristic.
\end{enumerate}
Considering the approximation ratio when solving the shrunk problem to optimality (solid lines),
we note that shrinking with correlations from the cycle relaxation (blue) \tb{always returns optimal} solutions.
This means that the inferred correlations \eqref{eq:cort} are indeed optimal.
\tb{Thus, if sub-optimal solutions to the original problem are retrieved when running QAOA instead of an exact algorithm for the shrunk problem, the degrade in solution quality must be attributed purely to QAOA.}
As expected, the potential decreases when shrinking randomly (red solid lines).
Here, correlations are sub-optimal which means that an optimal solution with the imposed correlations does not exist.

Now, we turn to the approximation ratio when solving the shrunk problem on the ideal quantum simulator (dashed-dotted lines).
As expected, when shrinking with optimal correlations inferred from the cycle relaxation (blue), the approximation ratio monotonically increases with vertex deletions.
A smaller sub-problem can be better approximated by the quantum algorithm.
Rather interestingly, an increase in approximation ratio when deleting more vertices can sometimes be observed even when shrinking randomly (red).
Here, the better approximability of the sub-problem by the quantum algorithm over-compensates the degrade in potential caused by shrinking sub-optimally.
\tb{These results highlight the strong limitation of QAOA at depth $p=1$.
  	Indeed, other studies suggest that $p\gtrapprox 11$ is required for QAOA to outperform classical algorithms, which is intractable for current quantum hardware platforms~\cite{lykov2022,Farhi_2022}.}

Presumably, the most interesting case is when the shrunk problem is solved on the real quantum machine (dashed lines).
The approximation ratio qualitatively follows the ideal simulation for both, optimal correlations (blue) and random shrinking (red).
However, as expected, the quantum hardware always performs worse than the ideal simulation.
Of course, this is due to noise effects in real hardware.
Notably, we observe a maximum approximation-ratio at 2 deleted vertices in both instances when shrinking randomly (red).
Here, the trade-off between potential degrade due to sub-optimal shrinking and performance gain due to increased approximability is optimal.
Finally, we note a significant advantage of QAOA over the coin-flipping-heuristic (dotted lines), when shrinking one or more vertices.
Without shrinking, \ie zero shrunk vertices, we observe that noise effects take over, and QAOA effectively flips a coin for every vertex.

\begin{figure}[]
	\centering
	\subfloat[Instance 2x3b.]{\includegraphics[width=0.7\linewidth]{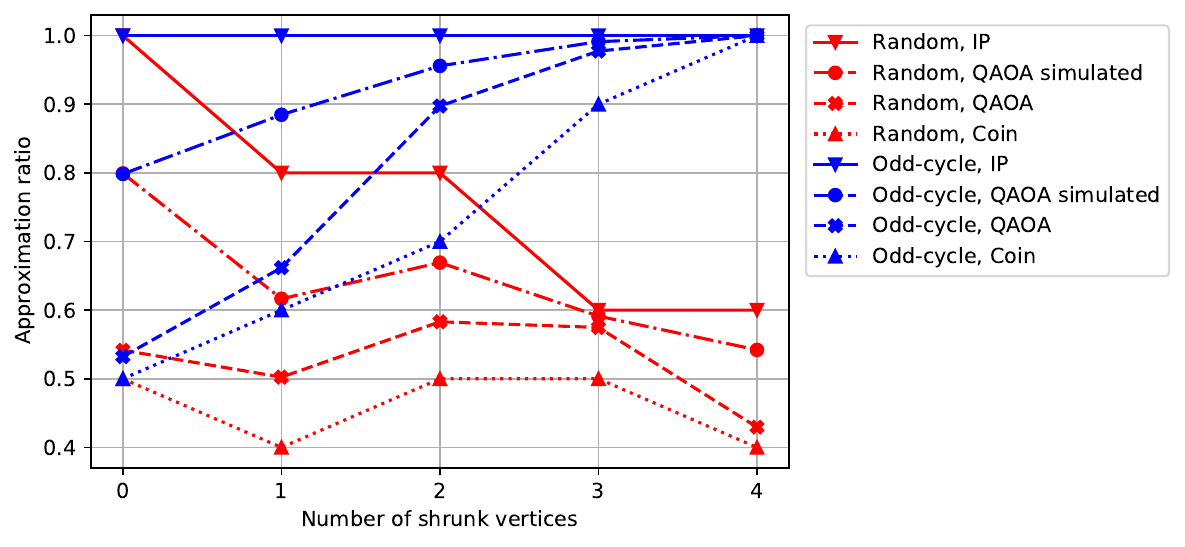}\label{fig:shrink_vs_val_2x3}}\\
	\subfloat[Instance k6b.]{\includegraphics[width=0.7\linewidth]{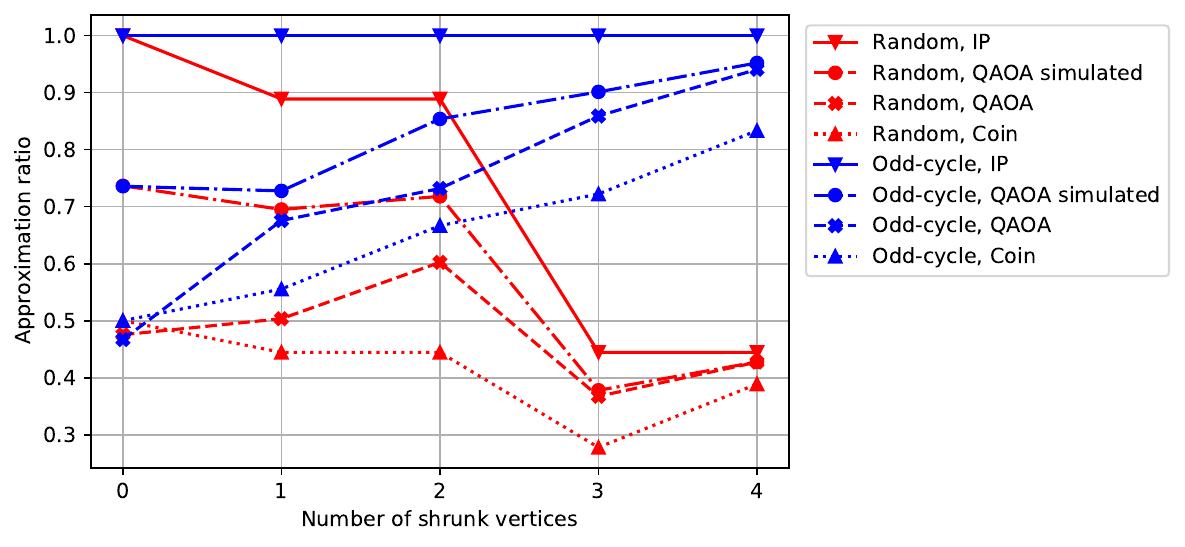}\label{fig:shrink_vs_val_sk6}}\\\
	\caption{Results for evaluating the shrinking-algorithm with different setting. Shown is the approximation ratio, defined in \eqref{eq:appratio}, versus the number of shrunk vertices.
		In the legend, \qm{Random} stands for random shrinking while \qm{Odd-cycle} stands for correlations given by \eqref{eq:cort}. \qm{IP}, \qm{QAOA simulated}, \qm{QAOA} and \qm{Coin} refer to different sub-problem solution methods as discussed in the main text body.}
	\label{fig:shrink_size}
\end{figure}

Corresponding figures for instances k6n and 3x3b appear in Appendix~\ref{app:exp_compl}.
The two key observations are the same as in Fig.~\ref{fig:shrink_size}: First, the
cycle relaxation yields optimal correlations.
Second, maxima in the approximation ratio for QAOA on real hardware are present when shrinking randomly.
We do not provide figures for the remaining instances from Table~\ref{tab:shrink}
because conclusions are similar to the previously discussed instances.

Summarizing, the results indicate that the proposed shrinking method
can indeed reduce problem size significantly without degrading solution quality:
In our experiments, shrinking with the cycle relaxation always
preserved optimality, \ie inferred correlations are optimal.
This is beneficial for quantum computation on noisy hardware of limited size.
When shrinking optimally, the performance of QAOA always increased with the number of deleted vertices.
\tb{From this, we conclude that QAOA at depth $p=1$ is strongly limited in finding high-quality solutions, even for small instances.}
Interestingly, when shrinking sub-optimally, we observed an optimal
trade-off between potential-loss and increased approximability of the reduced problem.
By combining linear programming with QAOA,
we were able to solve MaxCut instances from literature an order of magnitude larger
than the capabilities of current digital quantum hardware.
Although our experiments yield only a proof-of-principle, they encourage the combination
of quantum algorithms with classical branch-and-cut,
when QC approaches regimes where it outperforms classical heuristics.

\subsection{\tb{Comparison of LP Shrinking to RQAOA and Goemans-Williamson}}\label{sec:exp_rqaoa}
As mentioned earlier, any method for obtaining variable correlations can be substituted in the shrinking procedure from Section~\ref{sec:algo} by substituting Equation~\eqref{eq:cort} appropriately.
In particular, using correlations obtained by QAOA and recalculating these correlations after every shrinking step, recovers the well-known RQAOA~\cite{Bravyi_2020}.
In this section, we compare the quality of LP-based shrinking procedure to RQAOA.
To this end, we shrink random MaxCut instances via both, RQAOA and the LP-based method proposed in this work.

For a fair comparison, we recalculate the cycle relaxation after every shrinking step,
analogous to RQAOA, which continuously updates QAOA correlations.
QAOA correlations are obtained by evaluating~\eqref{eq:expth} at parameters estimated by~\eqref{eq:params}.
Although estimated parameters can be sub-optimal in principle, the results in Sec.~\ref{sec:exp_params} reveal a good performance of the estimated parameters  on random MaxCut instances, \cf Table~\ref{tab:param_pred}.

As instances, we choose the same Erdős–Rényi graphs as in Sec.~\ref{sec:exp_params}.
That is, we consider graphs on $100$ nodes with edge densities $d\in\set{0.05,0.1,0.15,0.2}$.
Again, the motivation is to cover densities below and above the critical density of $d^*=0.1$.
Considering runtimes, solving a single LP relaxation takes $\sim 60$~s while calculating RQAOA correlations is in the order of $~1$~s.

In Fig.~\ref{fig:rqaoa} we visualize the results.
Analogous to the previous section, we plot the number of shrunk vertices versus the average approximation ratio when solving the shrunk problem to optimality.
This allows us to compare the LP-based shrinking algorithm to RQAOA.
First, when considering low density instances with $d=0.05$ (blue lines in Fig.~\ref{fig:rqaoa}), we observe that LP-shrinking (dots) yields higher average approximation ratios than RQAOA (crosses) for all numbers of shrunk vertices.
LP shrinking constantly yields approximation ratios above 99.7 \% whereas the RQAOA approximation ratio decreases with increasing number of shrunk vertices below 96.5 \%.
For the critical density $d=0.1$ (orange) LP performs slightly worse than for $d=0.05$, while at the same time, RQAOA performs slightly better, thus reducing the advantage of LP.
However, LP still outperforms RQAOA for all numbers of shrunk vertices.
For the high density instances with $d=0.15$ (green), the performance of LP decreases further, in particular when shrinking more than half of the vertices.
At the same time, the performance of RQAOA increases compared to density $d=0.1$.
However, LP still outperforms RQAOA except when shrinking more than 70 vertices.
This trend continues.
For higher densities, the approximation ratio of LP shrinking decreases more rapidly with the number of shrunk vertices.
At the same time, RQAOA achieves larger approximation ratios for higher densities.
Still, at a density of $d=0.2$ (red), LP outperforms RQAOA up to shrinking 50 vertices.
Our results are in agreement with previous studies,
which revealed that the cycle relaxation is strong
on sparse instances~\cite{Charfreitag_2022,Rehfeldt_2022,Rendl2010,Bonato_2014}.

For completeness, we also compare RQAOA and LP-based shrinking to the
well-known classical Goemans-Williamson algorithm, which uses positive
semidefinite optimization.
To this end, we perform the randomized rounding algorithm 
for every instance and calculate the average approximation ratio for each ensemble (solid lines).
Both, LP-shrinking and RQAOA outperform the Goemans-Williamson algorithm for almost all densities and numbers of shrunk vertices.
The only exceptions are the high-density instances with $d=0.2$ (red) where Goemans-Williamson outperforms LP when shrinking more than 60 vertices.

In summary, we observe that linear programming yields high-quality
correlations for instances with densities of up to 20~\%, which
covers a very wide range of problem classes.
Even on denser instances, RQAOA outperforms LP only when shrinking more than half of the vertices.
Similar to Section~\ref{sec:exp_shrink}, we conclude that depth-1-QAOA often produces sub-optimal correlations
and larger QAOA depths are required to outperform classically obtained correlations.
However, larger QAOA depths complicates parameter optimization, potentially remedying the runtime advantage over classical methods.

\begin{figure}[]
	\centering
	\subfloat{\includegraphics[width=0.8\linewidth]{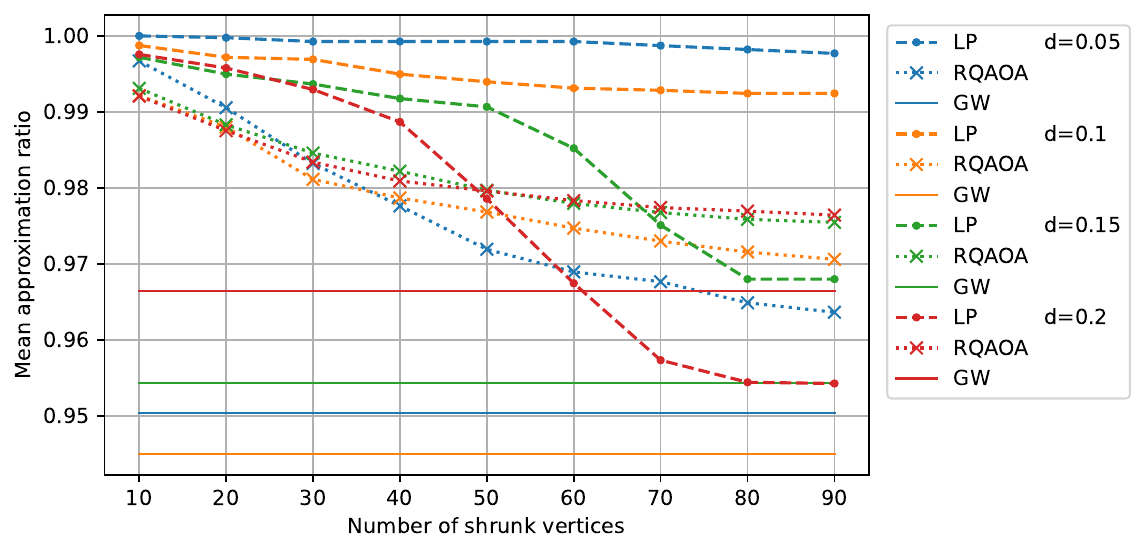}\label{fig:shrink_vs_val_2x3}}\\
	\caption{\tb{Comparison of LP-based shrinking to RQAOA and Goemans-Williamson. We plot the number of shrunk vertices versus the average approximation ratio. 
			Instances are Erdős–Rényi random graphs on 100 vertices with different densities $d\in \set{0.05,0.1,0.15,0.2}$.
			 Data points are averaged over 20 instances.
		 	For completeness, we also include the average approximation ratio for the Goemans-Williamson rounding algorithm (solid lines).}}
	\label{fig:rqaoa}
\end{figure}

	\section{Conclusion and Outlook}\label{sec:concl}
In this work, we proposed a hybrid quantum-classical heuristic algorithm for the maximum cut problem.
The guiding idea is to combine the ability of classical integer programming 
to solve large relaxations with the \tb{promise} of quantum optimization to find high-quality solutions quickly.
The algorithm is particularly well-suited for integration in classical integer programming for two reasons:
First, it relies on relaxation solutions and, second, it helps to avoid enumeration in classical branch-and-cut.
More specifically, we use the well-known technique of graph shrinking to reduce problem size such that
it can be handled by quantum hardware with limited resources.
To this end, we shrink according to an optimum of the cycle relaxation of the cut polytope. 
Furthermore, we improved the applicability
of \tb{depth-1} QAOA for weighted MaxCut, which builds the quantum part in the hybrid algorithm,
by deriving optimal parameters for instances on regular and triangle-free graphs
with weights following a binary distribution.
This result motivates a parameter estimate for arbitrary instances.

Our experiments give a proof-of-principle for the applicability of the proposed methods.
Although all considered instances can be solved in reasonable time
by purely classical algorithms,
the results indicate a potential benefit for integer programming, when
QC improves.
First, the proposed QAOA parameter estimate
works well in practice.
This improves the applicability of \tb{depth-1} QAOA since
it renders classical parameter optimization unnecessary.
Second, when combining shrinking with QAOA,
we observed that linear programming can shrink
problem size significantly without losing optimality.
\tb{In particular, linear programming outperforms depth-1 RQAOA on sparse instances.}
Furthermore, we observed that shrinking is indeed beneficial
for QAOA when executed on current quantum hardware.
\tb{Our results indicate that depth-1 QAOA is strongly limited in finding optimum solutions and quantum algorithms
of higher quality are needed for practical utility.}
Of course, a more thorough evaluation on a wider range of instances is needed
to investigate the performance of our method in more detail.

A direction of future research is the incorporation of characteristics
of quantum algorithms and quantum hardware in the process of shrinking.
It is known, that QAOA performs worse on certain types of graphs, \eg bipartite graphs.
From a hardware perspective, sparse graphs simplify the implementation of QAOA.
In the process of shrinking, one can try to avoid or produce such specific graph characteristics.
Moreover, other techniques for deriving (optimal) correlations
exist in literature. Their performance in our framework needs to be further studied.
\tb{Extending the set of test instances allows deeper insights into the instance-dependence of the proposed algorithm.}
Another field of ongoing research is the quantum part which
may be replaced by \tb{higher-depth} QAOA \tb{or variants like warm-start QAOA} or even by different algorithmic paradigms like quantum annealing.

	\bibliographystyle{unsrturl}  
	\bibliography{bib_cqaoa}
	\appendix
	\newpage
\section{Additional Figures}\label{app:exp}
\subsection{Parameter Estimate Evaluation}\label{app:exp_param}
\begin{figure}[h]
	\centering
	\subfloat[Instance d.]{\includegraphics[width=0.48\linewidth]{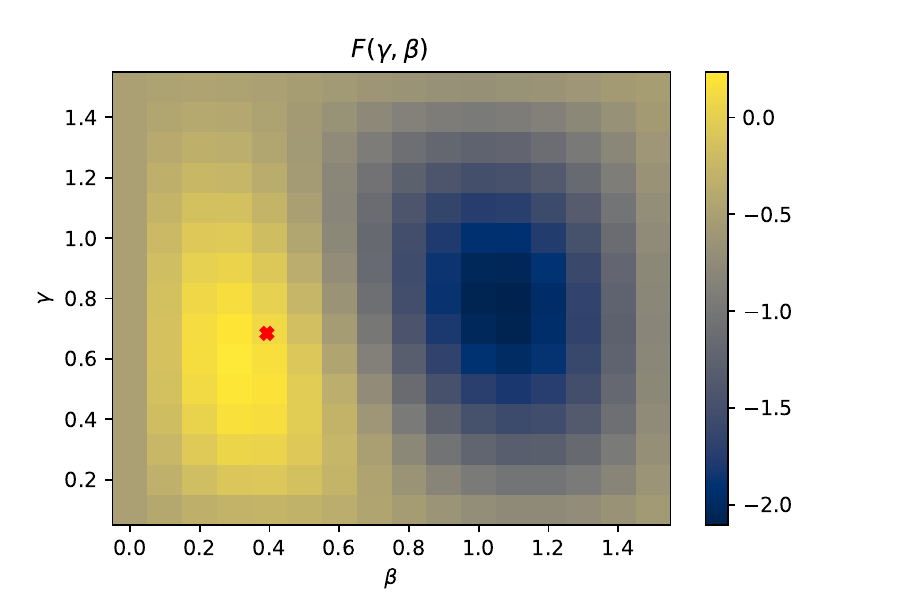}\includegraphics[width=0.48\linewidth]{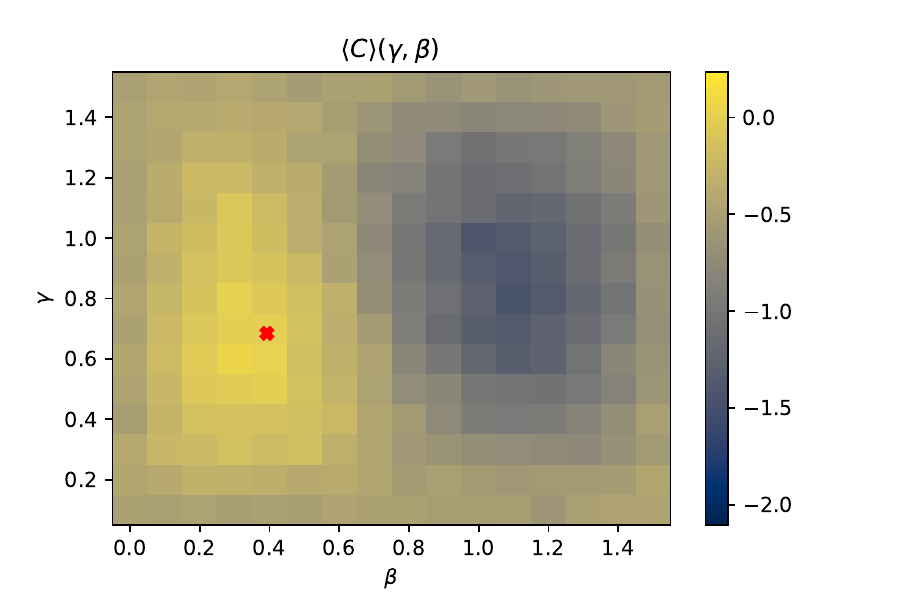}\label{fig:pred_chord}}\\
	\subfloat[Instance e.]{\includegraphics[width=0.48\linewidth]{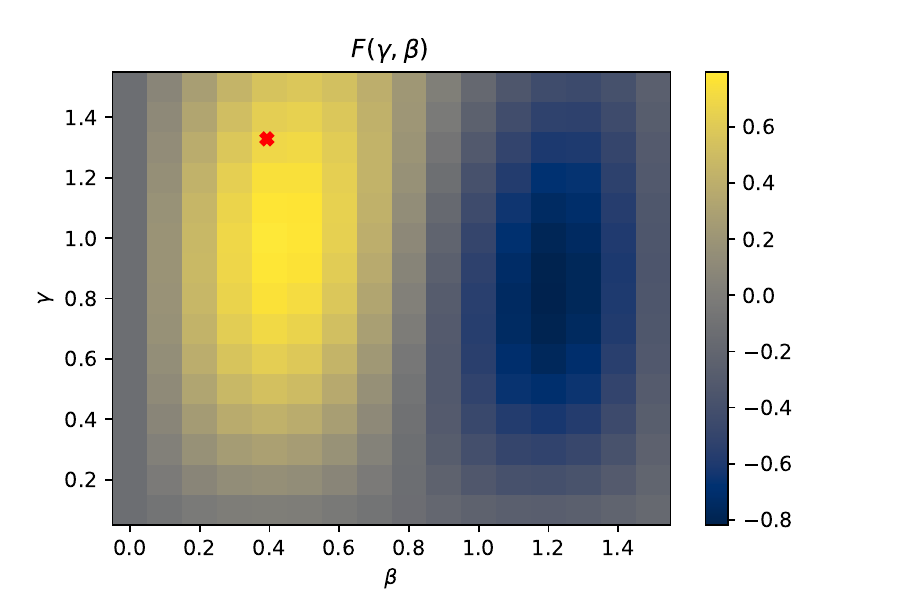}\includegraphics[width=0.48\linewidth]{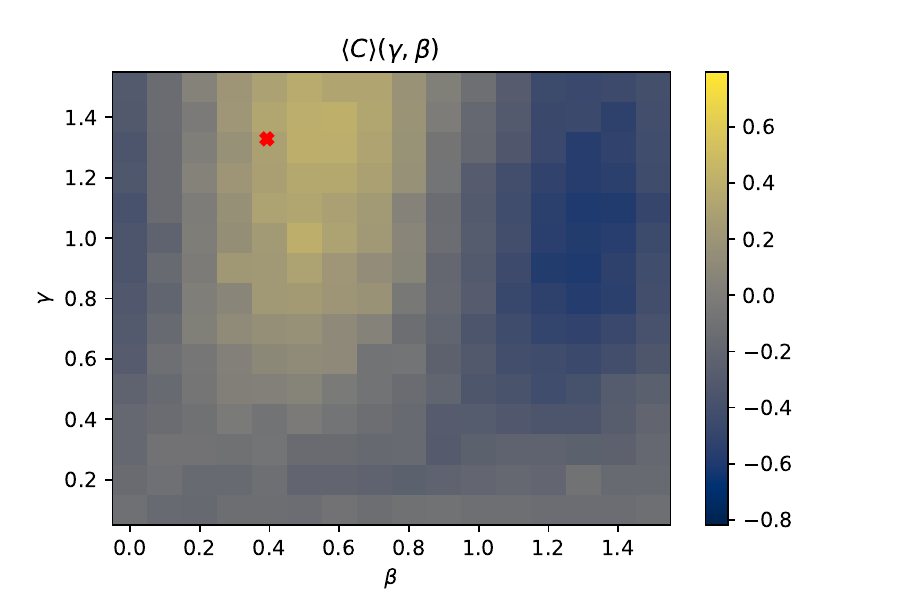}\label{fig:pred_compln}}\\

	\subfloat[Instance f.]{\includegraphics[width=0.48\linewidth]{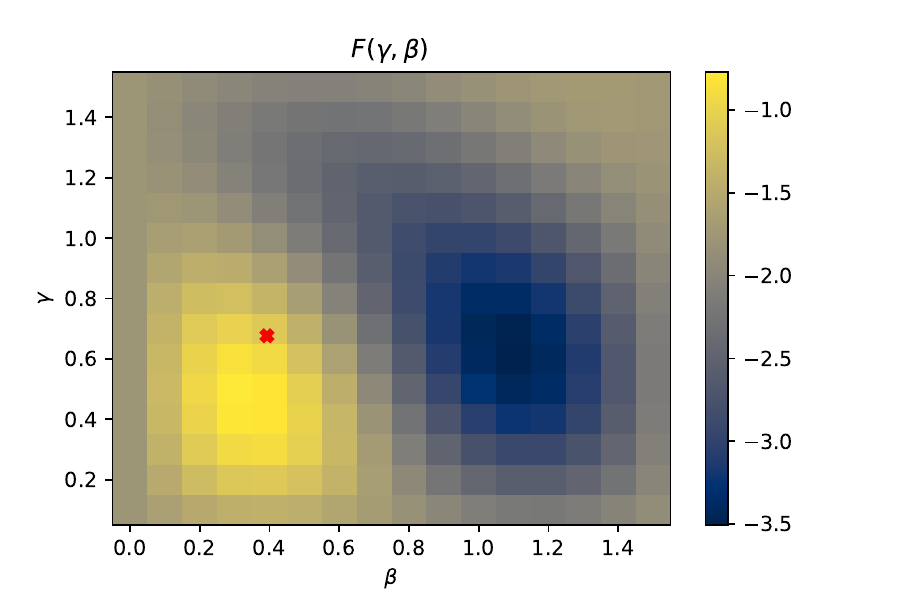}\includegraphics[width=0.48\linewidth]{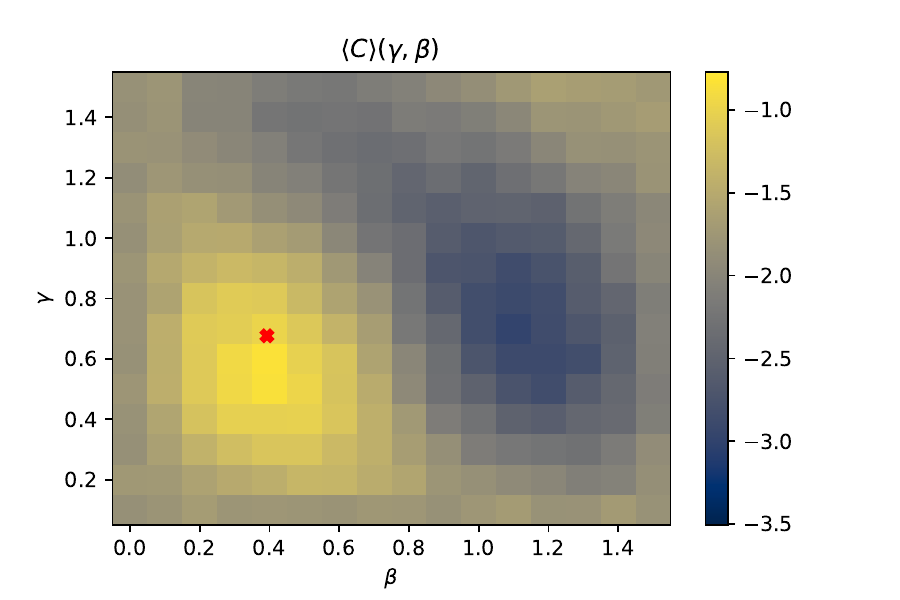}\label{fig:pred_chordn}}
	\caption{Visualized results for QAOA parameter-estimate on instances d, e and f from Fig.~\ref{fig:param_instances}. Compare Fig.~\ref{fig:pred}.}
	\label{fig:pred3}
\end{figure}
\newpage
\subsection{Combining Shrinking with QAOA}\label{app:exp_compl}
\begin{figure}[h]
	\centering
	\subfloat[Instance k6n from Table~\ref{tab:shrink}.]{\includegraphics[width=0.7\linewidth]{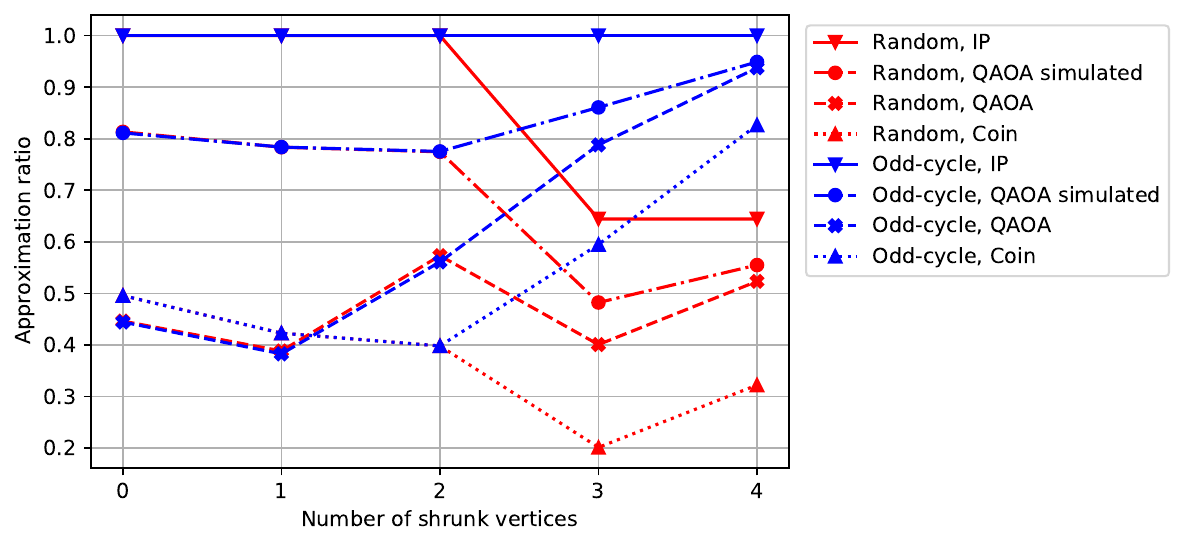}\label{fig:shrink_vs_val_sk6n}}\\\
	\subfloat[Instance 3x3b from Table~\ref{tab:shrink}.]{\includegraphics[width=0.7\linewidth]{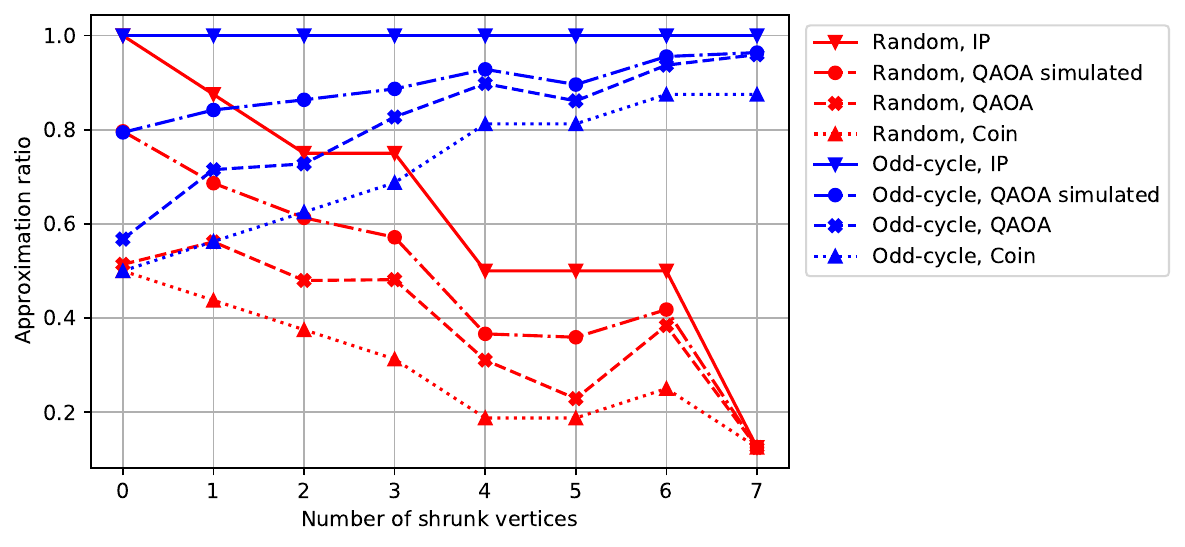}\label{fig:shrink_vs_val_3x3}}\\\
	\caption{Further results for evaluating the shrinking-algorithm with different setting. Compare Fig.~\ref{fig:shrink_size}.}
	\label{fig:shrink_size2}
\end{figure}

	\section*{Acknowledgements}
	This research has been supported
	by the Bavarian Ministry of Economic Affairs, Regional Development
	and Energy with funds from the Hightech Agenda Bayern and
	by the Federal Ministry for	Economic Affairs and Climate Action on the basis
	of a decision by the German Bundestag through project QuaST.
\end{document}